%% file: subsetsum.tex
\documentclass[USenglish,a4paper,11pt]{article}
\pdfoutput=1

\input{Style.sty}
\title{Classical and Quantum Algorithms for Variants of \\ Subset-Sum via Dynamic Programming}
\author{
  ~~~~~~~~
  Jonathan Allcock\thanks{Tencent Quantum Laboratory, Hong Kong. \url{jonallcock@tencent.com}} ~~~ \and
  Yassine Hamoudi\thanks{Simons Institute for the Theory of Computing, University of California, Berkeley. \url{ys.hamoudi@gmail.com}} ~~~ \and
  Antoine Joux\thanks{CISPA Helmholtz Center for Information Security. \url{joux@cispa.de}} ~~~~~~~~ \and
  Felix Klingelh\"ofer\thanks{G-SCOP, Université Grenoble Alpes. \url{felix.klingelhofer@grenoble-inp.fr}} \and
  Miklos Santha\thanks{Centre for Quantum Technologies and MajuLab, National University of Singapore. \url{miklos.santha@gmail.com}}
}
\date{\today}

\begin{document}

\maketitle


\begin{abstract}
  \input{abstract}
\end{abstract}

\newpage

\section{Introduction}
\label{sec:intro}
\input{introduction}

\section{Preliminaries}
\label{sec:Prelim}
\input{preliminaries}

\section{Dynamic programming data structure}
\label{Sec:DynProg}
\input{dynamic}

\section{\Sub}
\label{Sec:Subset}
\input{subset}

\section{\SSub}
\label{Sec:DSum}
\input{quantumShiftedSum}

\section{Pigeonhole variants of \Equ}
\label{sec:pigeonhole}
\input{pigeonhole}

\section{Acknowledgements}
\input{acknowledgements}

\printbibliography[heading=bibintoc]

\appendix

\section{Classical algorithms for \SSub}
\label{app:classical-shifted-sums}
\input{classicalShiftedSum}

\section{Quantum \Equ\ in terms of minimum solution ratio}
\label{app:quantum-equal}
\input{quantumEqualSum}

\end{document}

%% file: abstract.tex
\Sub\ is an NP-complete problem where one must decide if a multiset of $n$ integers contains a subset whose elements sum to a target value $m$. The best-known classical and quantum algorithms run in time $\wbo{2^{n/2}}$ and $\wbo{2^{n/3}}$, respectively, based on the well-known meet-in-the-middle technique. Here we introduce a novel classical dynamic-programming-based data structure with applications to \Sub\ and a number of variants, including \Equ\ (where one seeks two disjoint subsets with the same sum), \TSub\ (a relaxed version of \Sub\ where each item in the input set can be used twice in the summation), and \SSub, a generalization of both of these variants, where one seeks two disjoint subsets whose sums differ by some specified value.

Given any modulus $p$, our data structure can be constructed in time $\bo{n^2p}$, after which queries can be made in time $\bo{n^2}$ to the lists of subsets summing to any value modulo $p$. We use this data structure in combination with variable-time amplitude amplification and a new quantum pair finding algorithm, extending the quantum claw finding algorithm to the multiple solutions case, to give an~$\bo{2^{0.504n}}$ quantum algorithm for \SSub. This provides a notable improvement over the best-known~$\bo{2^{0.773n}}$ classical running time established recently by Mucha et al.~\cite{MNPW19c}. Incidentally, we obtain new~$\wbo{2^{n/2}}$ and~$\wbo{2^{n/3}}$ classical and quantum algorithms for \Sub, not based on the seminal meet-in-the-middle approach of Horowitz and Sahni~\cite{HS74j}. We also study \PEqu\ and \PMEqu, variants of \Equ\ where the existence of a solution is guaranteed by the pigeonhole principle. For the former problem, we give faster classical and quantum algorithms with running time $\wbo{2^{n/2}}$ and~$\wbo{2^{2n/5}}$, respectively. For the more general modular problem, we give a classical algorithm that also runs in time~$\wbo{2^{n/2}}$.

%% file: introduction.tex
\Sub\ is the problem of deciding whether a given multiset of $n$ integers has a subset whose elements sum to a target integer $m$.

\begin{problem}[\Sub]
  \label{Pbm:Sub}
  Given a multiset $\{a_1, \dots, a_n\}$ of positive integers and a target integer~$m$, find a subset $S \subseteq [n]$ such that $\sum_{i \in S} a_i = m$.
\end{problem}

It is often useful to express \Sub\ using inner product notation. We set  $\bar{a} = (a_1, \ldots , a_n) \in \N^n$, where the elements are taken in arbitrary order, and the task is to find $\bar{e} \in \{0,1\}^n$ such that
  $\bar{a} \cdot \bar{e} = \sum_{i=1}^n a_i e_i = m$.
The problem is famously NP-complete and featured on Karp's list of 21 NP-complete problems~\cite{Kar72c} in 1972 (under the name of knapsack). It can be solved classically in time $\wbo{2^{n/2}}$ via the \textit{meet-in-the-middle} technique~\cite{HS74j}. Whether this problem can be solved in time $\wbo{2^{\pt{1/2-\delta} n}}$, for some $\delta >0$, is an important open question, but we know that the Exponential Time Hypothesis implies that \Sub\ cannot be computed in time $m^{o(1)} 2^{o(n)}$~\cite{BLT15j, JLL16j}. \Sub\ can also be solved in pseudopolynomial time, for instance in $\bo{nm}$ by a textbook dynamic programming approach, which was improved to a highly elegant $\wbo{n + m}$ randomized algorithm by Bringmann~\cite{Bri17c}. However, assuming the Strong Exponential Time Hypothesis (\seth), it can be shown that for all~$\eps > 0,$ there exists $\delta > 0,$ such that \Sub\ cannot be computed in time $\bo{m^{1- \eps} 2^{n \delta}}$~\cite{ABHS19c}. On a quantum computer, meet-in-the-middle can be combined with quantum search to solve \Sub\ in time $\wbo{2^{n/3}}$. A modular version of \Sub\ can be similarly defined:

\begin{problem}[\MSub]
  \label{Pbm:MSub}
  Given a multiset $\{a_1, \dots, a_n\}$ of positive integers, a target integer $m$ and a modulus $q$, find a subset $S \subseteq [n]$ such that $\sum_{i \in S} a_i \equiv m \pmod q$.
\end{problem}

The $\wbo{2^{n/2}}$ classical and $\wbo{2^{n/3}}$ quantum meet-in-the-middle algorithms, as well as the classical
$\bo{nm}$ dynamic programming algorithm can be used to solve \MSub\ with the same running times, by replacing regular addition with modular addition. While the $\wbo{n + m}$ algorithm of Bringmann does not immediately give rise to an $\wbo{q}$ algorithm for \MSub, several recent algorithms have achieved this complexity~\cite{ABJ19c, ABB21c, CI21c}. Also, SETH implies that for all $\eps > 0,$ there exists $\delta > 0,$ such that \MSub\ cannot be computed in time $\bo{q^{1- \eps} 2^{n \delta}}$ because an instance of \Sub\ where each $a_i < m$ is a special case of \MSub\ when we choose $q = nm$.


\subsection{Some variants of \Sub}

\Sub\ has several close relatives we will be concerned with in this paper. First among these is \Equ, introduced by Woeginger and Yu~\cite{WY92j}, where one must decide if a set of $n$ positive integers contains two disjoint subsets whose elements sum to the same value:

\begin{problem}[\Equ~\cite{WY92j}]
  \label{Pbm:Equ}
  Given a set $\{a_1, \dots, a_n\}$ of positive integers, find two distinct subsets $S_1, S_2 \subseteq [n]$ such that $\sum_{i \in S_1} a_i = \sum_{i \in S_2} a_i$. In inner product notation, we are looking for a nonzero vector $\bar{e} \in \{-1, 0,1\}^n$ such that
$\bar{a} \cdot \bar{e} = 0.$
\end{problem}

The folklore classical algorithm~\cite{Woe08j} for \Equ\ runs in time $\wbo{3^{n/2}}\le\bo{2^{0.793n}}$, and is also based on a meet-in-the-middle approach. In the classical case, we arbitrarily partition the input into two sets of the same size, giving rise to vectors $\bar{a}_1, \bar{a}_2 \in \N^{n/2}$. Then we compute and sort the possible $3^{n/2}$ values $\bar{a}_1 \cdot \bar{e}$, for $\bar{e} \in \{-1, 0,1\}^{n/2}$. Finally, we compute the possible $3^{n/2}$ values of the form $\bar{a}_2 \cdot \bar{e}$ and, for each value, check via binary search if it has a collision (i.e. an item of the same value) in the first set of values. In the quantum case, we use a different balancing, dividing the input into a set of size $n/3$ and a set of size $2n/3$, and then use quantum search over the larger set to find a collision. This folklore quantum algorithm has a running time of $\wbo{3^{n/3}}\le\bo{2^{0.529n}}$. The classical running time of \Equ\ was reduced in a recent work by Mucha et al.~\cite{MNPW19c} to $\bo{2^{0.773n}}$, and it is an open problem whether this can be further improved. The modular version of \Equ\ is defined as:

\begin{problem}[\MEqu]
  \label{Pbm:MEqu}
  Given a set $\{a_1, \dots, a_n\}$ of positive integers and a modulus $q$, find two distinct subsets $S_1, S_2 \subseteq [n]$ such that $\sum_{i \in S_1} a_i \equiv \sum_{i \in S_2} a_i \pmod q$.
\end{problem}

Similarly to \Sub, the $\bo{2^{0.793n}}$ time meet-in-the-middle algorithm for \Equ\ gives rise to an algorithm of the same time for \MEqu. Moreover, we can suppose that $q \geq 2^n$, because otherwise we can just consider $a_1, \ldots , a_k$ from the input, where~$k$ satisfies $2^{k-1} \leq q < 2^k$. By the pigeonhole principle, such an instance has a solution which we will show (see Theorem~\ref{thm:phPMSub}) can be found in time $\wbo{2^{k/2}}$. Thus, \MEqu\ can always be solved in time $\bo{q^{0.793}}$ classically. Intriguingly, when expressed as a function of~$n$, faster algorithms are known for both \Sub\ and \MSub\ than for \Equ\ and \MEqu, respectively, whereas expressed as a function of $q$ (or as a function of $\sum_{i=1}^n a_i$ in the non-modular cases), the situation is the opposite. This holds both classically and quantumly.

A natural generalization of \Sub\ is to allow each item in the input set to be used more than once in the summation, where the maximum number of times each item can be used is specified as part of the input to the problem.
This is the analog of bounded knapsack, a well-studied problem in the literature (see for example~\cite{KPP04b}). In particular, we will study the case when every item can be used at most twice.

\begin{problem}[\TSub]
  \label{Pbm:TSub}
  Given a multiset $\{a_1, \dots, a_n\}$ of positive integers and a target integer $0 < m < 2 \sum_{i=1}^n a_i$, find a vector $\bar{e} \in \{0,1, 2\}^n$, such that $\bar{a} \cdot \bar{e} = m$.
\end{problem}

There is a natural variant of \Sub\ that generalizes both  \Equ\ and \TSub. We call this variant \SSub, whose investigation is the main subject of this paper.

\begin{problem}[\SSub]
  \label{Pbm:SSub}
  Given a multiset $\{a_1, \dots, a_n\}$ of positive integers and an integer $0 \leq s < \sum_{i=1}^n a_i$, find two distinct subsets $S_1, S_2 \subseteq [n]$ such that  $\sum_{i \in S_1} a_i = s + \sum_{i \in S_2} a_i$.
\end{problem}

The condition $S_1 \neq S_2$ is necessary in the case $s = 0$ to exclude the trivial solutions $S_1 = S_2$.
The problem \Equ\ is a special case of \SSub\ in this case, and it is easy to show (see Proposition~\ref{Prop:size}) that \TSub\ can also be reduced to \SSub\ without increasing the size of the input. This means that any algorithm for \SSub\ automatically gives rise to an algorithm of the same complexity for \Equ\ and \TSub, and therefore we focus on constructing classical and quantum algorithms for \SSub. We also consider the modular version of \SSub:

\begin{problem}[\MSSub]
  \label{Pbm:MSSub}
  Given a multiset $\{a_1, \dots, a_n\}$ of positive integers, an integer $0 \leq s < \sum_{i=1}^n a_i$ and a modulus $q$, find two distinct subsets $S_1, S_2 \subseteq [n]$ such that $\sum_{i \in S_1} a_i \equiv s + \sum_{i \in S_2} a_i \pmod q$.
\end{problem}

We additionally study the following variant of \Equ\ where, by the pigeonhole principle, a solution is guaranteed to exist. This search problem is total in the sense that its decision version is trivial because the answer is always `yes'. Such problems belong to the complexity class TFNP~\cite{MP91j} consisting of NP-search problems with total relations. Problems in TFNP cannot be NP-hard unless NP equals co-NP. More precisely, the following two problems belong to the Polynomial Pigeonhole Principle complexity class PPP, defined by Papadimitriou~\cite{Pap90c}, where the totality of the problem is syntactically guaranteed by the pigeonhole principle.

\begin{problem}[\PEqu]
  \label{Pbm:PEqu}
  Given a set $\{a_1, \dots, a_n\}$ of positive integers such that $\sum_{i=1}^n a_i < 2^n-1$, find two distinct subsets $S_1, S_2 \subseteq [n]$ such that $\sum_{i \in S_1} a_i = \sum_{i \in S_2} a_i$.
\end{problem}

There are $2^n$ subsets $S \subseteq [n]$. Since they all verify $0 \leq \sum_{i \in S} a_i \leq 2^n-2$ there must exist two distinct subsets $S_1,S_2$ that sum to the same value, according to the pigeonhole principle. The modular version of \PEqu\ similarly belongs to the class PPP:

\begin{problem}[\PMEqu]
  \label{Pbm:PMEqu}
  Given a set $\{a_1, \dots, a_n\}$ of positive integers and a modulus $q$ such that $q \leq 2^n - 1$, find two distinct subsets $S_1, S_2 \subseteq [n]$ such that $\sum_{i \in S_1} a_i \equiv \sum_{i \in S_2} a_i \pmod q$.
\end{problem}

Observe that \PEqu\ is a special case of \PMEqu\ when $q = 2^n - 1.$


\subsection{Our contributions and techniques}

We give new classical and quantum algorithms for \Sub\ and several closely related problems defined in the previous section. Our results are succinctly stated below and summarized in Table~\ref{tb:results-summary}. The algorithms for \Sub\ achieve the same complexity as the currently best-known algorithms based on the meet-in-the-middle method\footnote{After completion of this work, it was pointed out to us by an anonymous referee that a classical algorithm for \Sub, similar to ours, was sketched in~\cite{AKKN16c}.}. Our quantum algorithm for \SSub\ (and for its special cases of \Equ\ and \TSub) improves on the currently best-known $\bo{2^{0.529n}}$ quantum algorithm for these problems, which is also based on meet-in-the-middle. Our quantum algorithm for \PEqu\ further improves, in this special case, on our algorithm for general \Equ. We also initiate the study of the \PEqu\ problem (and its modular variant) in the classical setting, where we obtain a better complexity than what was known before for the general \Equ\ problem.

\begin{rtheorem}[Theorems~\ref{thm:Quantum-DP-SS},~\ref{thm:Classical-DP-SS} {\normalfont (Restated)}]
  There are representation-technique-based classical and quantum algorithms for \Sub\ that run in time~$\wbo{2^{n/2}}$ and~$\wbo{2^{n/3}}$, respectively.
\end{rtheorem}

\begin{rtheorem}[Theorems~\ref{Thm:Dsum},~\ref{Thm:classical-Dsum} {\normalfont (Restated)}]
  There are classical and quantum algorithms for \SSub\ that run in time~$\bo{2^{0.773 n}}$ and $\bo{2^{0.504n}}$, respectively.
\end{rtheorem}

\begin{rtheorem}[Theorem~\ref{thm:phsq} {\normalfont (Restated)}]
  There is a quantum algorithm for \PEqu\ that runs in time~$\wbo{2^{2n/5}}$.
\end{rtheorem}

\begin{rtheorem}[Theorems~\ref{thm:phsAlgo},~\ref{thm:phPMSub} {\normalfont (Restated)}]
There are classical deterministic algorithms for \PEqu\ and \PMEqu\ that run in time~$\wbo{2^{n/2}}$.
\end{rtheorem}

\begin{table}[htb]
  \centering
     \renewcommand{\arraystretch}{1.4}
  \begin{tabular}{c@{\hskip 2mm}cc}
       \toprule
              & Classical & Quantum  \\ \midrule
       \Sub   & $2^{n/2}$~\cite{HS74j,AKKN16c}, [Thm.~\ref{thm:Classical-DP-SS}]& $2^{n/3}$~\cite{BJLM13c}, [Thm.~\ref{thm:Quantum-DP-SS}] \\
       \SSub  & $2^{0.773n}$~\cite{MNPW19c}, [Thm.~\ref{Thm:classical-Dsum}] & $2^{0.504n}$~[Thms.~\ref{Thm:Dsum},~\ref{Thm:Esum_quantum}]   \\
       \PEqu  & $2^{n/2}$~[Thm.~\ref{thm:phsAlgo}] & $2^{2n/5}$~[Thm.~\ref{thm:phsq}]\\[1.4mm]
       {\renewcommand{\arraystretch}{0.9} \begin{tabular}{@{}c@{}} \textsc{Pigeonhole Modular} \\ \textsc{Equal-Sums} \end{tabular}} & $2^{n/2}$~[Thm.~\ref{thm:phPMSub}] & \textendash \\ \bottomrule
  \end{tabular}
  \caption{Best-known classical and quantum running times for variants of \Sub. Our results are indicated by reference to the corresponding theorems in this paper. The $\wbo{\cdot}$ notation is implied for all running times. For \Sub\ our dynamic programming based results have the same time complexity as the best previous algorithms that were based on meet-in-the-middle. The results of~\cite{MNPW19c} are for \Equ, a special case of \SSub. In Appendix~\ref{app:classical-shifted-sums}, we give an extension of their algorithm to \SSub.} \label{tb:results-summary}
\end{table}

At a high level, all of our algorithms use a \emph{representation technique} approach. While this technique was originally designed to solve \Sub\ when the instances are drawn from some specific distribution~\cite{HJ10c}, here we follow the path of Mucha et al.~\cite{MNPW19c} and use it in a worst-case analysis. Among our three main algorithms, the quantization of this technique for \SSub\ is the most challenging. We will therefore explain first, via this algorithm, the difficulties we had to address and the methods we used to tackle them.

\subparagraph*{\SSub.}
The representation technique approach for \SSub\ consists first of selecting a random prime $p\in \{2^{bn}, \ldots , 2^{bn+1}\}$, where $b \in (0,1)$ is some appropriate constant, and a random integer $k \in  \{0, \ldots , p-1\}$. Then we consider the random {\em bin}~$T_{p,k}$, defined as
  \[T_{p,k} = \set[\Big]{S\subseteq \{1, \ldots , n\} : \sum_{i \in S} a_i \equiv k \pmod p},\]
and we search that bin and $T_{p,(k-s) \bmod p}$ for a colliding solution (i.e. a pair of sets $(S_1,S_2) \in T_{p,k} \times T_{p,(k-s) \bmod p}$ such that $\sum_{i \in S_1} a_i = s + \sum_{i \in S_2} a_i$). The choice of the bin size (which, on average, is roughly~$2^{(1-b)n}$) should balance two opposing requirements: the bins should be sufficiently large to contain a solution and also sufficiently small to keep the cost of collision search low.

To satisfy the above two requirements, our algorithm uses the concept of a \textit{maximum solution}. This is the maximum of $\abs{S_1} + \abs{S_2}$, when $S_1, S_2$ are disjoint and form a solution.  Let this maximum solution size be~$\ell n$, for some $\ell\in(0,1)$. The algorithm consists of two different procedures, designed to handle different maximum solution sizes. For $\ell$ close to 0 or close to 1, the quantization of the meet-in-the-middle method adapted to solutions of size $\ell n$ is used because it performs better. In this case, the quantization does not present any particular difficulties: it is a straightforward application of quantum search with the appropriate balancing. We therefore focus the discussion on the representation technique procedure used for values of~$\ell$ away from $0$ or $1$. When $S_1, S_2$ form a maximum solution of size $\ell n$ then, for every set $X \subseteq \overline{S_1 \cup S_2}$ in the complement of the solution, the pairs $S_1 \cup X, S_2 \cup X$ also form a solution, and all these solutions have different values (see Lemma~\ref{lem:maxRatioEsum}). This makes it possible to bound from below, not only the number of solutions, but also the number of {\em solution values} by $2^{(1 - \ell) n}$, which makes the use of the representation technique successful.

The most immediate way to quantize the procedure is to replace classical collision finding with the quantum element distinctness algorithm of Ambainis~\cite{Amb07j}. However, in a straightforward application of this algorithm we face a difficulty. For concreteness, we explain this when $\ell = 3/5$. In that case, by the above, the total number of solutions with different values is at least~$2^{2n/5}$. This is handy for applying quantum element distinctness: we can select a random prime $p \in \{2^{2n/5}, \ldots , 2^{2n/5 + 1}\}$ and expect to have a solution in the random bin $T_{p,k} $ with reasonable probability. The expected size $|T_{p,k}|$ of the bin is about~$2^{3n/5}$, and therefore the running time of Ambainis' algorithm should be of the order of $|T_{p,k}|^{2/3}$ which is also about $2^{2n/5}$. However, the quantum element distinctness algorithm requires us to perform queries to $T_{p,k}$. That is, for some indexing $T_{p,k} = \{S_1, \ldots, S_{\abs{T_{p,k}}}\}$ of the elements of~$T_{p,k}$, we need to implement the oracle
  $O_{T_{p,k}}\ket{I}\ket{0} = \ket{I}\ket{S_I},$
where $1 \leq I \leq \abs{T_{p,k}}$. In other words, given $1 \leq I \leq \abs{T_{p,k}}$, we have to be able to find the $I$th element in~$T_{p,k}$ (for some ordering of that set). In the usual description of  the element distinctness algorithm there is a simple way to do that (for example, the set over which the algorithm is run is just a set of consecutive integers). However, finding a simple bijection among the first~$|T_{p,k}|$ integers and~$T_{p,k}$ is not a trivial task. Unlike in the classical case, explicitly enumerating $T_{p,k}$ is not an option because this would take too long, requiring about~$2^{3n/5}$ time steps. Instead, we use \emph{dynamic programming} to compute the table of cardinalities,
  \[t_p[i,j] = \abs[\Big]{\set[\Big]{S \subseteq \{1,\dots,i\} : \sum_{s \in S} a_s \equiv j \pmod p}}.\]
Computing the cardinality of the bins is cheaper than computing their contents, and can be done
in time $\bo{n^2 p} = \wbo{2^{2n/5}}$. Crucially, once the table is constructed, one can deduce the paths through it that led to $t_p[n,k] = \abs{T_{p,k}}$, in order to find each element of $T_{p,k}$ in time $\bo{n^2}$. More precisely, we define a particular strict total order $\prec$ over ${\cal P}([n])$ and prove:

\begin{rtheorem}[Theorem~\ref{Thm:effEnum} {\normalfont (Restated)}]
  Let $T_{p,k}$ be enumerated as $T_{p,k} = \{S_1, \ldots, S_{\abs{T_{p,k}}}\}$ where $S_1 \prec \dots \prec S_{\abs{T_{p,k}}}$. Given any integer $I \in \set{1, \dots, \abs{T_{p,k}}}$ and random access to the elements of the table~$t_p$, the set $S_I$ can be computed in time~$\bo{n^2}$.
\end{rtheorem}

This novel data structure will be used in our algorithms for \Sub, \SSub\ and \PEqu. We now describe the additional quantum tools we use for \SSub. The algorithm randomly chooses a bin of size about $2^{(1-b)n}$ where $b$ is defined differently depending on whether~$\ell$ is above or below~3/5, as different quantum tools are required in these two regions. When $\ell \leq 3/5$, with high probability a random bin will contain multiple solutions from which we can profit. To that end, we construct a quantum algorithm for finding a pair marked by a binary relation $\rel(x,y) := \ind{f(x) = g(y)}$ that tests if two values $f(x)$ and $g(y)$ are equal or not. Our algorithm generalizes the quantum element distinctness~\cite{Amb07j} and claw finding~\cite{Tan09j} algorithms to the case of multiple marked pairs. Using an appropriate variant of the birthday paradox (see Lemma~\ref{Lem:BPclaw}) we prove:

\begin{rtheorem}[Theorem~\ref{Thm:PairFinding} {\normalfont ({\sc Quantum pair finding} - Restated)}]
  Consider two sets of $N \leq M$ elements, respectively, and an evaluation function on each set.
  Suppose that there are $K$ disjoint pairs in the product of the two sets such that in each pair the elements evaluate to the same value. There is a quantum algorithm that finds such a pair in time
  $\wbo[\big]{(NM/K)^{1/3}}$ if $N \leq M \leq KN^2$ and $\wbo[\big]{(M/K)^{1/2} }$ if $M \geq KN^2$.
\end{rtheorem}

The best complexity when $\ell \leq 3/5$ is then obtained by choosing the bin size parameter~$b$ as a function of $\ell$, which balances the cost of the construction of the dynamic programming table and the quantum pair finding. When $\ell > 3/5$, choosing a bin size $2^{(1-b)n}$, for $b \le 1-\ell$, guarantees that a random bin contains at least one solution with high probability. However, a better running time at first seems to be achievable by the following argument: Choose $b > 1-\ell$, for which there is an exponentially small probability that a random bin contains a solution, and use amplitude amplification to boost the success probability. Balancing again the dynamic programming and quantum pair finding costs would then give an optimal bin size of $2^{3n/5}$, independent of $\ell$. However, this argument contains a subtlety. Standard amplitude amplification requires that the random bin $T_{p,k}$ simultaneously satisfies two conditions: besides containing a solution, it should also have size close to the expected size of about $2^{(1-b)n}$. But there is no guarantee that these two events coincide, and a priori it could be that the exponentially small fraction of $T_{p,k}$ containing a solution also happens to have sizes that far exceed the expectation. Fortunately, by carefully bounding the expectation of the product of bin sizes, we can use the variable-time amplitude amplification algorithm of Ambainis~\cite{Amb12c}, and achieve the same running time as given by the above argument. We believe that this is a nice and natural application of this method. The running time of our algorithm for \SSub, as a function of $\ell$, is shown in Fig.~\ref{fig:quantum-ess-improved}.

\begin{figure}[htbp]
  \centering
  \includegraphics[width=0.65\textwidth]{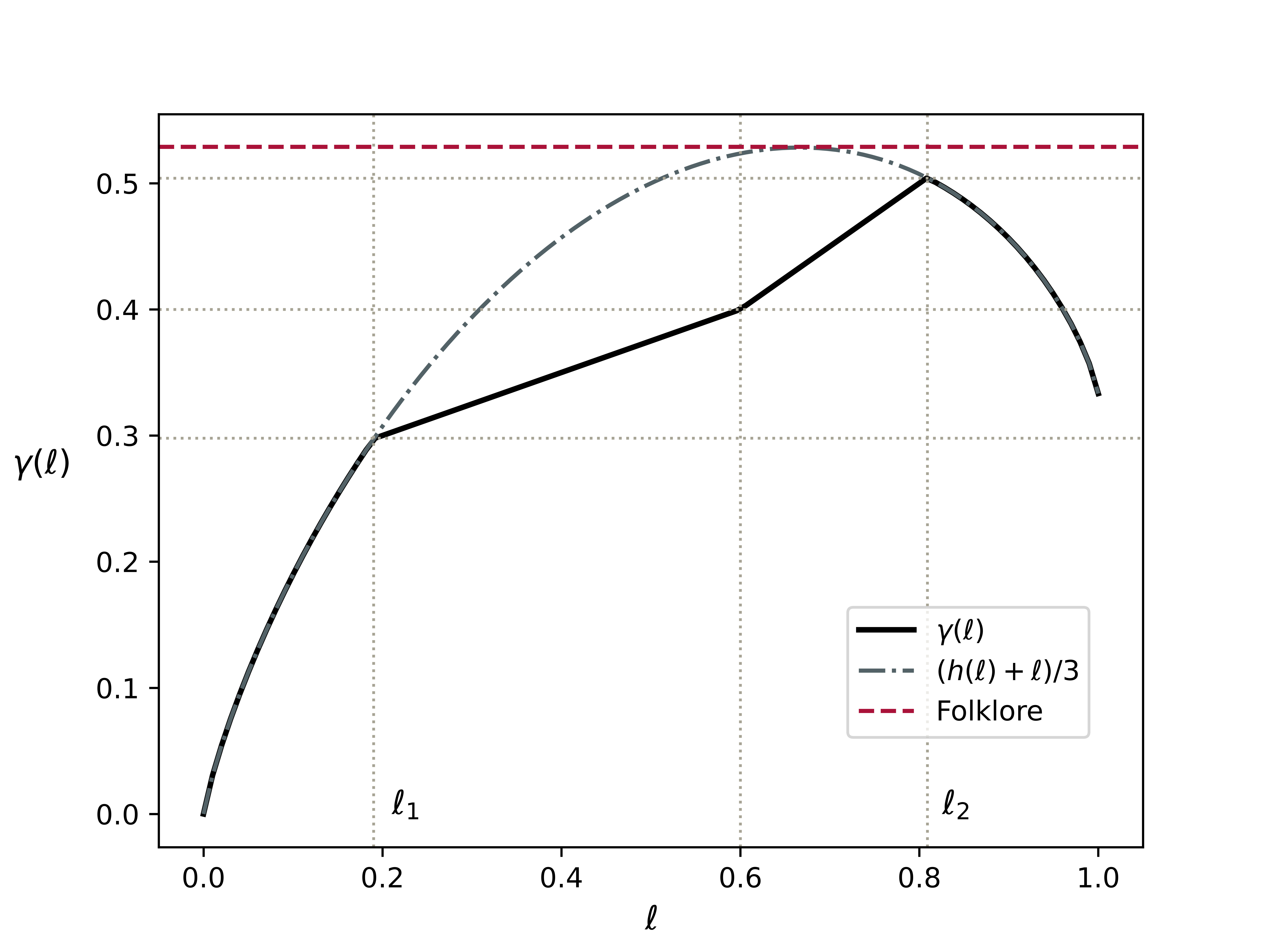}
  \caption{Running time exponent $\gamma(\ell)$ of the quantum \SSub\ algorithm, as a function of the maximum solution ratio $\ell$ (see Theorem~\ref{Thm:Dsum}). The maximum value of $\gamma(\ell)$ is $\approx 0.504$ and it occurs at $\ell=\ell_2\approx 0.809$. For reference, the curve $(h(\ell) + \ell)/3$ corresponding to Theorem~\ref{Thm:Mim-Dsum} is plotted for all values of $\ell$, as is the value of $0.529$ corresponding to the exponent of the folklore (quantized) meet-in-the-middle algorithm, as applied to \SSub.}
  \label{fig:quantum-ess-improved}
\end{figure}

\subparagraph*{\PEqu.}
This problem can be solved by any (classical or quantum) algorithm that solves the general \Equ\ (or \SSub) problem. However, one can make use of the explicit promise of $a_1 + \dots + a_n < 2^n-1$ to design faster algorithms than provided for by the general case when $\ell > 3/5$. Indeed, by the pigeonhole principle, for \emph{any} value of $p$, if a bin $T_{p,k}$ has a size larger than $2^n/p$ then it must contain a solution (see Lemma~\ref{lem:pigeonhole-collision}). Moreover, there must exist at least one such oversized bin. The array~$t_p$ can now be constructed both for \emph{locating the index $k$} of one oversized bin and searching for a solution in it. We thus obtain a classical algorithm running in time $\wbo{p + 2^n/p}$, and a quantum algorithm running in time $\wbo{p + (2^n/p)^{2/3}}$. These two quantities are minimized by \emph{deterministically} choosing $p = 2^{n/2}$ and $p = 2^{2n/5}$ respectively (see Section~\ref{sec:non-mod}).

\subparagraph*{\PMEqu.}
In general, we do not know how to extend our techniques to the problems modulo some integer~$q$. A natural approach would be to consider the bins $T_{p,k} = \set*{S\subseteq \{1, \ldots , n\} : \pt{\sum_{i \in S} a_i \bmod q} \equiv k \pmod p}$, but it is unclear how to compute the corresponding table $t_p$ efficiently. We give a solution to this problem for \PMEqu\ that works in the classical setting only (see Section~\ref{Sec:PMAlgo}).


\subsection{Related works}

The closest work to our contribution is the paper of Mucha et al.~\cite{MNPW19c} solving \Equ\ classically in time $\bo{2^{0.773n}}$. Their algorithm and ours use the same two basic procedures, based respectively on the meet-in-the-middle method and the representation technique. Let us point out some of the differences. Unlike our algorithm that is based on the concept of the size of a maximum solution, the classical algorithm is analysed as a function of a {\em minimum size} solution, defined as $\abs{S_1} + \abs{S_2}$, where this sum is minimized over all solutions. The use the classical algorithm makes from a minimum solution $S_1, S_2$ of size~$\ell n$ is that when $\ell > 1/2$ the number of {\em solution values} can be bounded from below by $2^{(1 - \ell) n}$. However, this  does not hold for \TSub\ when $\ell n$ is the size of a minimum solution, but \textit{is} valid for both \Equ\ and \TSub\ when it is the size of a maximum solution. Another difference with~\cite{MNPW19c} is that their classical representation technique  algorithm always samples~$p$ from the same set $\set{2^{(1 - \ell)n}, \dots,  2^{(1 - \ell)n + 1}}$, while we randomly choose $p\in \set{2^{bn}, \dots, 2^{bn +1}}$ where $b$ is defined differently depending on in which of two distinct regions $\ell$ lies. This makes it possible to use different quantum techniques in these two regions.

The representation technique was designed by Howgrave-Graham and Joux~\cite{HJ10c} to solve random \Sub\ instances under some hypotheses (heuristics) about how such instances behave during the run of the algorithm. The idea is to decompose a single solution to the initial problem into many distinct decompositions of a sum of half-solutions. To compensate for this blow-up, an additional linear constraint is added to select approximately one of these decompositions. Under some rather strong assumptions, which are satisfied for a large fraction of randomly chosen instances, \cite{HJ10c} can solve \Sub\ instances in time $\bo{2^{0.337n}}$. Since then, several variants of this classical method have been proposed~\cite{BCJ11c, EM20c, BBSS20c, CJRS22c}, while others have investigated quantum algorithms based on the representation technique. Bernstein et al.~\cite{BJLM13c} improved on~\cite{HJ10c} using quantum walks, and their algorithm (again under some hypotheses) runs in time $\bo{2^{0.242n}}$. Further quantum improvements were made in this context by~\cite{HM18c} and~\cite{BBSS20c}. However, we emphasize that the algorithms in all these papers work for random inputs generated from some distributions. The paper~\cite{MNPW19c} gave the first classical algorithm based on the representation technique that works for worst-case inputs and with proven bounds. To our knowledge, for worst-case inputs with provable guarantees, the first quantum algorithm based on the representation technique is given in our work.

Dynamic programming is notoriously hard to quantize, with a key obstacle being the intrinsically sequential way in which the solution to a large problem is constructed from the solutions to smaller subproblems. Certain basic dynamic programming algorithms can be trivially accelerated by quantum search or minimum finding (see, e.g.~\cite{Abb19c}) but beyond that few other quantum improvements are known. One notable exception is the work of Ambainis et al.~\cite{ABI19c} who gave faster quantum algorithms for several NP-hard problems for which the best classical algorithms use dynamic programming. Their algorithms precompute solutions for smaller instances via dynamic programming and then use non-trivial quantum search recursively on the rest of the problem. In our work, the dynamic programming subroutine that we use is classical (although for \SSub\ it is performed in superposition) and the sequential nature of the process is therefore not an issue. Rather than using quantum computing to accelerate classical dynamic programming, we instead use dynamic programming to enable fast queries, required for quantum search and pair finding, to be performed on complicated sets.


The class PPP is arguably less studied than other syntactically definable subclasses of TFNP, such as PLS (Polynomial Local Search) and PPA (Polynomial Parity Argument), and it is not known whether \PEqu\ is complete in PPP. In fact, the first complete problem for the class was only identified relatively recently~\cite{SZZ18c}. Our results for \PEqu\ suggest that the problem is indeed simpler to solve than \Equ. In spirit, a similar result was obtained in~\cite{BST02j} where it was shown that, for an optimization problem closely related to \Equ, better approximation schemes can be obtained for instances with guaranteed solutions.


\subsection{Open problems}

We suggest two directions for future work on the modular versions of the problems studied in this paper:

\begin{enumerate}
  \item Quantization of \PMEqu. That is, can we construct a quantum algorithm improving on the classical running time given in Theorem~\ref{thm:phPMSub}?
  \item Can we prove a modular version of the algorithm of~\cite{MNPW19c}? That is, can \MEqu\ be solved classically in time $\bo{q^{0.773}}$?
\end{enumerate}


\subsection{Structure of the paper}

The paper is organised as follows. In Section~\ref{sec:Prelim} we define the quantum computational model and some basic notations, and we state several facts, propositions and algorithmic tools (such as our quantum pair finding algorithm) used in subsequent sections. In Section~\ref{Sec:DynProg} we introduce the dynamic programming data structure and show how it can be used to implement fast subset-sum queries. We present a simple application of this data structure to the \Sub\ problem in Section~\ref{Sec:Subset}, where we obtain new classical and quantum algorithms achieving the same complexity as the best-known algorithms based on meet-in-the-middle. The main applications are described in Section~\ref{Sec:DSum} for the \SSub\ problem, and in Section~\ref{sec:pigeonhole} for the pigeonhole variants of \Equ. Finally, we adapt our results to the classical \SSub\ problem in Appendix~\ref{app:classical-shifted-sums}, and we describe an alternative quantum algorithm for \Equ\ whose complexity is parametrized by the minimum solution ratio in Appendix~\ref{app:quantum-equal}.

%% file: preliminaries.tex
\subsection{Notations}

We use the $\wbo{x}$ and $\wom{x}$ notations to hide factors that are polylogarithmic in the argument~$x$. For integers $0 \leq m <n$, we denote by $[m \tdots n]$ the set $\{m, m+1, \ldots n\}$, and by $[n]$ the set~$[1 \tdots n]$. For sets $S,S' \subseteq [n]$ we denote by $\bar{S}$ the complement $[n] \setminus S$, and by $S \Delta S'$ the symmetric difference of~$S$ and~$S'$. Given a multiset $A = \{a_1, \ldots, a_n\}$ and a subset $S\subseteq [n]$ we denote $\Sigma_A(S) := \sum_{i\in S} a_i$. When the set $A$ is clear from the context, we will omit the subscript and simply denote the subset sum by $\Sigma(S)$. The power set of $[n]$  will be denoted by ${\cal P}([n]) := \{S: S \subseteq[n]\}$. For arbitrary integers $a$ and $b$ and a modulus~$p$ we say that $a$ is congruent to $b$ modulo $p$, and we write $a \equiv b \pmod p$ or $a \equiv_p b$ if $a - b$ is divisible by~$p$. By $a \bmod p$ we denote the unique integer in $\{0, \ldots , p-1\}$ that is congruent to $a$ modulo~$p$. The binary entropy function will be denoted by $h(x) = -x\log_2(x) - (1-x)\log_2(1-x)$.


\subsection{Quantum computational model}

Similar to previous works on quantum element distinctness~\cite{Amb07j}, quantum dynamic programming~\cite{ABI19c} and quantum subset sum algorithms~\cite{BJLM13c}, in our quantum algorithm running time analysis we assume the standard circuit model (where computational time corresponds to the number of single and two-qubit gates) augmented with random access to quantum memory. That is, coherent access to any element of an $m$-qubit array can be performed in time polylogarithmic in $m$. Note that fully quantum memory is required in Algorithm~\ref{Algo:Rep-Dsum} for \SSub\ since multiple bins $T_{p,k}$ must be computed and stored in superposition.  On the other hand,  Algorithm~\ref{Algo:Quantum-DP-SS} for \Sub\ only requires a single bin $T_{p,k}$ to be searched. Thus, while the memory cells must be accessed in superposition, the data that each cell holds is classical.


\subsection{Basic Facts}

\begin{fact}[\cite{Gal68b}, page 530]
  \label{fact:entropy}
  For every constant $\ell \in (0,1)$ and for every large enough integer $n$, the following bounds hold:
    $\frac{2^{n h(\ell)}}{\sqrt{8 n \ell(1-\ell)}}
  \leq {\binom{n}{\ell n}} <  \frac{2^{n h(\ell)}}{\sqrt{2\pi n \ell(1-\ell)}}$.
\end{fact}

\begin{fact}
  \label{fact:primedivide}
  Let $b>0$ be a constant, $n$ a large enough integer and $a_1 \neq a_2$ two integers. Then, for a random prime $p \in [2^{bn} \tdots 2^{bn +1} ]$ we have $\Pr_p\bc{a_1 \equiv_p a_2} \leq \frac{\log\pt{\abs{a_1-a_2}}}{2^{bn}}$.
\end{fact}

\begin{proof}
  The number of primes that belong to the interval $[2^{bn} \tdots 2^{bn +1} ]$ is at least $2^{bn}/bn$ for $n$ large enough (see \cite[p.371]{HW75b}). Moreover, there are at most $\frac{\log\pt{\abs{a_1-a_2}}}{bn}$ prime numbers larger than $2^{bn}$ that divide $a_1 - a_2$. The result follows by a union bound.
\end{proof}


\subsection{Size preserving reductions}

We say that a polynomial time reduction between two problems is {\em size preserving} if it preserves the number of input items. The following propositions justify the assumption in our algorithms that the input items are bounded by $2^{4n}$.

\begin{proposition}
  \label{Prop:size}
  There exist size preserving reductions from \Equ\ to \SSub\ and from \TSub\ to \SSub.
\end{proposition}

\begin{proof}
  In fact, the claimed reductions are not only size preserving, but also keep the input items $\{a_1, \ldots, a_n\}$ intact. \Equ\ is simply the restriction of \SSub\ to the case $s=0$. Consider an instance of \TSub, and set $W = \sum_{i=1}^n a_i$. We can suppose, without loss of generality, that $W < m < 2W$ because $\sum_{i=1}^n a_i \cdot e_i = m$ if and only if $\sum_{i=1}^n a_i \cdot (2 - e_i )= 2W -m$, and the $m=W$ case is trivial.

  We claim that $\bar{a} = (a_1, \ldots, a_n), m$ is a positive instance of \TSub\ if and only if $\bar{a}, m-W$ is a positive instance of \SSub. For a vector $\bar{e} \in \{0,1, 2\}^n$, define $e_i^1 = 1$ if $e_i = 2$ and 0 otherwise, and similarly define $e_i^2 = 1$ if $e_i = 0$ and 0 otherwise. Then $\bar{e}^1, \bar{e}^1 \in \{0,1\}^n$ and  $e_i = e_i^1 - e_i^2 +1$ implying $\bar{a} \cdot \bar{e} = m$ if and only if $\bar{a} \cdot \bar{e}^1 = \bar{a} \cdot \bar{e}^2 + m - W.$ Therefore $\bar{e}$ is a solution of \TSub\ if and only if $\bar{a}, m-W$ is a solution of \SSub.
\end{proof}

Given a positive instance of \TSub, the above reduction produces a positive instance of  \SSub\ with $s \neq 0$.
It is easy to check that the reduction also works in the reverse direction. Taken together, this implies that \SSub\ with $s \neq 0$ is just a reformulation of \TSub\ and, as we have already observed, with $s = 0$ is exactly \Equ.

\begin{proposition}
  \label{Prop:selfsize}
  There exists a probabilistic size preserving reduction from \SSub\ to an instance of \SSub\ where the input items satisfy $\sum_{i=1}^n a_i < 2^{4n}$. A similar statement holds also for \Sub.
\end{proposition}

\begin{proof}
  We prove the statement for \SSub, the proof for \Sub\ is analogous. Let $a_1, \dots, a_n, s$ be an instance of \SSub. Without loss of generality, we can suppose that $\sum_{i=1}^n a_i < 2^{2^n}$ because otherwise we can solve the instance in polynomial time of the input size. We choose a random prime $p \in \{2^{4n-1}, \ldots , 2^{4n} \}$ and define the reduced instance by $s' = s \bmod p$ and $a'_i = a_i \bmod p$. It is obvious that if the original instance has a solution then so does the reduced instance. We now claim that if the original instance does not have a solution, then the reduced instance has a solution only with probability at most $\bo{2^{-n}}$. This is because a random prime $p$ divides $\sum_{i \in S_1} a_i -  \sum_{i \in S_2} a_i - s$ with probability $\bo{\log(2^{2^n})/2^{4n}}$ when $(S_1,S_2)$ is not a solution (Fact~\ref{fact:primedivide}), and there are $\bo{2^{2n}}$ such pairs to consider.
\end{proof}


\subsection{Quantum algorithms}

We use the following generalization of Grover's search to the case of an unknown number of solutions.

\begin{fact}[{\sc Quantum search}, Theorem 3 in \cite{BBHT98j}]
  \label{fact:search}
  Consider a function $f : [N] \rightarrow \{0,1\}$ with an unknown number $K = \abs{f^{-1}(1)}$ of marked items. Suppose that~$f$ can be evaluated in time~$\tau$. Then, the quantum search algorithm finds a marked item in $f$ in expected time~$\wbo{\sqrt{N/K} \cdot \tau}$.
\end{fact}

Given a classical subroutine with stopping time $\tau$ that returns a marked item with probability~$\rho$, we can convert it into a constant success probability algorithm with expected running time $\bo{\ex{\tau}/\rho}$ by repeating it $\bo{1/\rho}$ times. Ambainis proved a similar result for the case of quantum subroutines, with a dependence on the second moment of the stopping time $\tau$, and a Grover-like speed-up for the dependence on $\rho$.

\begin{fact}[{\sc Variable-time amplitude amplification}, Theorem 2 in~\cite{Amb12c}]
  \label{fact:vtaa}
  Let $\mathcal{A}$ be a quantum algorithm that looks for a marked element in some set.
  Let $\tau$ be the random variable corresponding to the stopping time of the algorithm, and let $\rho$ be its success
  probability. Then the variable-time amplitude amplification algorithm finds a marked element in the above set with
  constant success probability in maximum time $\wbo{\sqrt{\ex{\tau^2}/\rho}}$.
\end{fact}

The next result is a variant of the Birthday paradox over a product space $[N] \times [M]$, where at least $K$ \emph{disjoint} pairs are marked by some binary relation $\rel$. Two pairs $(x,y)$ and $(x',y')$ are said to be disjoint if $x \neq x'$ and $y \neq y'$. The disjointness assumption is made to simplify the analysis and will be satisfied in our applications.

\begin{lemma}[\sc Variant of the Birthday paradox]
  \label{Lem:BPclaw}
   Consider three integers $1 \leq K \leq N \leq M$.
   Let~$\rel : [N] \times [M] \ra \rn$ be a binary relation such that there exist at least $K$ mutually disjoint pairs $(x_1,y_1),\dots,(x_K,y_K) \in [N] \times [M]$ with $\rel(x_k,y_k) = 1$ for all $k \in [K]$. Given an integer $r \leq \bo{\sqrt{NM/K}}$, define $\eps(r)$ to be the probability of obtaining both elements from at least one marked pair when~$r$ numbers from~$[N]$ and~$r$ numbers from $[M]$ are chosen independently and uniformly at random. Then,
      $\eps(r) \geq \om*{\frac{r^2K}{NM}}$.
\end{lemma}

\begin{proof}
  Fix any $K$ disjoint marked pairs $(x_1,y_1),\dots,(x_K,y_K)$. Let $X_1,\dots,X_r$ (resp. $Y_1,\dots,Y_r$) be~$r$ independent and uniformly distributed random variables over $[N]$ (resp. $[M]$). For any indices $i,j$, let $Z_{i,j}$ denote the binary random variable that takes value $1$ if $\{X_i,X_j\}$ is one of the $K$ fixed pairs, and set $Z = \sum_{i,j} Z_{i,j}$. By definition, we have $\eps(r) \geq \pr{Z \neq 0}$. We lower bound this quantity by using the inclusion-exclusion principle,
    \[\eps(r) \geq \sum_{i,j} \Pr(Z_{i,j} = 1) - \frac{1}{2} \sum_{(i,j) \neq (k,\ell)} \Pr(Z_{i,j} = 1 \wedge Z_{k,\ell} = 1).\]
  The first term on the right-hand side is equal to $\sum_{i,j} \Pr(Z_{i,j} = 1) = \frac{r^2K}{NM}$. For the second term, the analysis depends on whether the indices $i,k$ and $j,\ell$ are distinct or not. If they are distinct then $\Pr(Z_{i,j} = 1 \wedge Z_{k,\ell} = 1) = \frac{K^2}{(NM)^2}$ since $Z_{i,j}$ and $Z_{k,\ell}$ are independent. Otherwise, suppose for instance that~$i = k$. Since the $K$ fixed pairs are disjoint, we have $\Pr(Z_{i,j} = 1 \wedge Z_{i,\ell} = 1) = \Pr(Z_{i,j} = 1 \wedge Y_j = Y_{\ell}) = \frac{K}{NM^2}$. Finally, there are $4\binom{r}{2}^2$ ways of choosing the indices $i,j,k,\ell$ when $i \neq k$ and $j \neq \ell$, and $4r\binom{r}{2}$ ways when $i=k$ or $j=\ell$. By putting everything together we obtain that,
    $\eps(r) \geq \frac{r^2K}{NM} - 2\binom{r}{2}^2 \frac{K^2}{(NM)^2} - 2r\binom{r}{2}\frac{K}{NM^2} \geq \om*{\frac{r^2K}{NM}}$.
\end{proof}

We use the above result to construct a quantum algorithm for finding a marked pair when the relation $\rel(x,y)$ is determined by checking if two underlying values $f(x)$ and $g(y)$ are equal or not. Our analysis essentially generalizes the quantum element distinctness~\cite{Amb07j} and claw finding~\cite{Tan09j} algorithms to the case of $K > 1$.

\begin{theorem}[\sc Quantum pair finding]
  \label{Thm:PairFinding}
  There is a bounded-error quantum algorithm with the following properties. Consider four integers $1 \leq K \leq N \leq M \leq R$ with $R \leq N^{\bo{1}}$. Let $f : [N] \ra [R]$ and $g : [M] \ra [R]$ be two functions that can be evaluated in time $\tau$. Define~$\rel : [N] \times [M] \ra \rn$ to be any of the two following binary relations:
    \begin{enumerate}
        \item $\rel(x,y) = 1$ if and only if $f(x) = g(y)$.
        \item $\rel(x,y) = 1$ if and only if $f(x) = g(y)$ and $x \neq y$.
    \end{enumerate}
  Suppose that there exist at least $K$ mutually disjoint pairs $(x,y) \in [N] \times [M]$ such that $\rel(x,y) = 1$. Then, the algorithm returns one such pair in time
    \[\left\{\begin{array}{ll}
      \wbo[\big]{(NM/K)^{1/3} \cdot \tau} & \text{\rm if} \ N \leq M \leq KN^2, \\ [3mm]
      \wbo[\big]{(M/K)^{1/2} \cdot \tau} & \text{\rm if} \ M \geq KN^2.
    \end{array}\right.\]
\end{theorem}

\begin{proof}
  If $N \leq M \leq KN^2$ the algorithm consists of running a quantum walk over the product Johnson graph $J(N,r) \times J(M,r)$ with $r = (NM/K)^{1/3}$. This walk has spectral gap $\delta = \om{1/r}$ and the fraction $\eps$ of vertices containing both elements from at least one marked pair satisfies $\eps \geq \om*{\frac{r^2K}{NM}}$ by Lemma~\ref{Lem:BPclaw}. Using the MNRS framework~\cite{MNRS11j}, the query complexity of finding one marked pair is then $\bo{S + \frac{1}{\sqrt{\eps}} (\frac{U}{\sqrt{\delta}} + C)}$, where the setup cost is $S = r$, the update cost is $U = \bo{1}$, and the checking cost is $C = 0$. This leads to a query complexity of $\bo[\big]{r + 1/\sqrt{\eps \delta}} = \bo{(NM/K)^{1/3}}$. By a simple adaptation of the data structures described in~\cite[Section 6.2]{Amb07j} or \cite[Section 3.3.4]{Jef14d}, this can be converted to a similar upper bound on the time complexity with a multiplicative overhead of $\tau$.

  If $M \geq KN^2$, the algorithm instead stores all pairs $\set{(x,f(x))}_{x \in [N]}$ in a table -- sorted according to the value of the first coordinate -- and then runs the quantum search algorithm on the function $F : [M] \ra \rn$ where $F(x') = 1$ if there exists $x \in [N]$ such that $\rel(x,y) = 1$. There are at least~$K$ marked items and~$F$ can be evaluated in time $\bo{\tau + \log N}$ using the sorted table. Thus, the running time is $\wbo{N \cdot \tau + (M/K)^{1/2} \cdot (\tau + \log N)} = \wbo{(M/K)^{1/2}}$ by Fact~\ref{fact:search}.
\end{proof}

%% file: dynamic.tex
Here we introduce our dynamic programming data structure and show how it can be used to implement low-cost queries to the elements of the set $T_{p,k}$ defined as follows.

\begin{definition}
  \label{def:ds}
  Let $A=\{a_1, \ldots, a_n\}$ be a multiset of $n$ integers. For integers $p \ge 2$ and $k\in\{0,1, \ldots, p-1\}$, define the set $T_{p,k}$ by
    $T_{p,k} = \{ S \subseteq[n] : \Sigma_A(S) \equiv k \pmod p\},$
  and denote the cardinality of $T_{p,k}$ by $t_{p,k}:= \abs{T_{p,k}}$.
\end{definition}

Our main tool is the table $t_p$, defined below, constructed by dynamic programming. As $t_{p,k} = t_p[n,k]$, once the table is constructed, the size $t_{p,k}$ of $T_{p,k}$ can be read off the last row.

\begin{lemma}
  \label{lem:dp-table-time}
  Let $n,p$ be non-negative integers. Define the $(n+1) \times p$ table $t_p[i,j] = \abs{\set{S \subseteq \{1,\dots,i\} : \Sigma(S) \equiv j \pmod p}}$ where $i \in [0 \tdots n]$ and $j \in [0 \tdots p-1]$. Then, $t_p$ can be constructed in time $\bo{n^2 p}$ by dynamic programming.
\end{lemma}

\begin{proof}
  To compute the elements of the table, observe that $t_p[0,0] = 1$ and $t_p[0,j] = 0$ for $j > 0$. The remaining elements can be deduced from the relation
    $t_p[i,j] = t_p[i-1,j] + t_p[i-1, (j-a_{i}) \bmod p]$.
  The $i^{th}$ row of $t_p$ can thus be deduced from the $(i-1)^{th}$ row and $a_{i}$ in time $\bo{np}$ (the $n$ factor is due to the entries being exponentially large) and the computation of all rows can be completed in time $\bo{n^2 p}$.
\end{proof}


\subsection{Construction of the Subset-Sum oracle}
\label{subsec:fast-sub}

We now show how to use the table $t_p$ to quickly query any element of $T_{p,k}$. To do so, we first define an ordering of the elements of $T_{p,k}$.

\begin{definition}
  Let $\prec$ be the relation over ${\cal P}([n])$ defined as follows: for all $S_1, S_2 \subseteq [n]$, $S_1 \prec S_2$ if and only if $\max\{i : i \in S_1 \Delta S_2\} \in S_2$.
\end{definition}

\begin{lemma}
  The relation $\prec$ is a strict total order.
\end{lemma}

\begin{proof}
  For every subset $S \subseteq [n]$, we define $\chi(S) = \sum_{i \in S} 2^i$. Then, $S_1 \prec S_2$ if and only if $\chi(S_1) < \chi(S_2)$. Since $<$ is a total order over the integers, so is $\prec$ over~${\cal P}([n])$.
\end{proof}

Using the above relation, we now show that the query function $f : [1 \tdots t_{p,k}] \rightarrow T_{p,k}$, defined by $f(I) = S_I$, can be computed in time~$\bo{n^2}$.

\begin{theorem}
  \label{Thm:effEnum}
  Let $T_{p,k}$ be enumerated as $T_{p,k} = \{S_1, \ldots, S_{t_{p,k}}\}$ where $S_1 \prec \dots \prec S_{t_{p,k}}$.
  Given any integer $I \in [1 \tdots t_{p,k}]$ and random access to the elements of the table $t_p$, the set $S_I$ can be computed in time~$\bo{n^2}$.
\end{theorem}

\begin{proof}
  Algorithm~\ref{Algo:PHS} gives a process that starts from $t_p[n,k]$ (i.e. the total number of subsets $S \subseteq [n]$ that sum to $k$ modulo $p$) and an empty set $Z$, and constructs $S_I$ by going backwards ($i=n,\dots,1$) through the rows of $t_p$. At the $i$-th step we examine~$t_p[i-1,j]$ and decide whether to include $i$ in $Z$ or not. If we do include $i$ then we examine another element in that row to decide a new value of $I$, and we also reset $j$.

  \begin{figure}[htbp]
    \begin{algorithm}[H]
      \label{Algo:PHS}
      \caption{Fast Subset-Sum oracle}
      \KwIn{Table $t_p$, integers $k \in [0 \tdots p-1]$ and $I \in [1 \tdots t_{p,k}]$.}
      \KwOut{The $I$th subset $Z \subseteq [n]$ (according to $\prec$) such that $\Sigma(Z) \equiv k \pmod p$.} \vspace{0.3\baselineskip}
      Set $j = k$ and $Z = \varnothing$. \\
      \For{$i = n, \dots, 1$}{
            \uIf{$I \leq t_p[i-1,j]$}{
              Do nothing.
            }
            \Else{
              Update $Z = Z \cup \{i\}$, $I = I - t_p[i-1,j]$ and $j = j - a_i \mod p$. \\
            }
        }
        Return $Z$.
    \end{algorithm}
  \end{figure}

  The algorithm consists of $n$ iterations, each of which can be performed in time $\bo{n}$ assuming random access to the (exponentially large) elements of $t_p$, and therefore the running time is~$\bo{n^2}$. What is left to prove is that the output of the algorithm is indeed~$S_I$.

  We first provide a high-level explanation of why the algorithm works. The total ordering defined by~$\prec$ implies that $T_{p,k}$ can be written as the disjoint union of two sets,
    $T_{p,k} = \{S_1, \ldots, S_{t_p[n-1,k]}\}\cup \{S_{t_p[n-1,k]+1}, \ldots, S_{t_p[n,k]}\},$
  where $n$ is not contained in any $S_i$ in the first (left) set, and is contained in every $S_i$ of the second (right) set. Thus, we add $n$ to the working set $Z$ only if $I > t_p[n-1,k]$.  If this is the case, $S_I$ is the $I-t_p[n-1,k]$-th element of the right set.  We note that removing $n$ from each $S_i$ in the right set gives the next bin defined over a smaller universe of size $n-1$,
    $\set*{S \subseteq [n-1] : \Sigma(S) \equiv k - a_n \pmod p}$
  that has $t_p[n-1, (k-a_n) \bmod p]$ elements. Therefore, by updating $n\leftarrow n-1$, $I \leftarrow I-t_p[n-1,k]$ and $k\leftarrow (k - a_n) \bmod p$ we can repeat the process to determine whether to add $n-1$ to the working set, and so on, until we reach the value $1$.

  More formally, denote by $I_i, j_i$ and $Z_i$ the values of the respective variables at the beginning of the $i$th iteration. With this notation we initially have $I_n = I, j_n =k$ and $Z_n = \varnothing$, and the final output corresponds to $Z_0$. We prove by backwards induction for $i = n, \dots, 1$ the following two statements, which clearly hold for $i=n$:
  \begin{enumerate}
     \item $Z_i = S_I \cap [i+1 \tdots n]$,
     \item $I_i \leq t_p[i, j_i]$ and in the enumeration of $T_p[i, j_i] := \{S_1, \ldots S_{t_p[i,j_i]}\} $ according to $\prec$ we have $S_{I_i} = S_I \cap [1 \tdots i]$.
  \end{enumerate}
  Note that the first statement implies that the final output is $Z_0 = S_I$.
  To prove the inductive step for the first statement, we must show that $Z_{i-1} = S_I \cap [i \tdots n]$. The inductive hypothesis implies that this holds exactly when $i \in Z_{i-1}  \Leftrightarrow i \in S_I$. Observe that $T_p[i-1,j_i] = \{S \in T_p[i,j_i] : i \not \in S \}$. Therefore the set $T_p[i,j_i]$ is the distinct union of $T_p[i-1,j_i]$ and~${\cal F}$ where ${\cal F} = \{S \in T_p[i,j_i] : i  \in S \}$. By the definition of $\prec$ over $T_p[i,j_i]$, for every $X \in T_p[i-1,j_i] $ and  for every $Y \in {\cal F},$ we have $X \prec Y$. This implies that $i \in S_{\ell} \Leftrightarrow t_p[i-1,j_i] < \ell$, for every $\ell \in [1 \tdots t_p[i,j_i]]$. Therefore we have the following equivalences:
  $
     i \in Z_{i-1}
      \Leftrightarrow  t_p[i-1,j_i] < I_i
      \Leftrightarrow  i \in S_{I_i}
     \Leftrightarrow  i \in S_I \cap [1 \tdots i]
     \Leftrightarrow  i \in S_I,$
  where the first equivalence follows from the definition of $Z_{i-1}$ and the third equivalence from the second statement of the inductive hypothesis.

  We now prove the inductive step for the second statement. Let the enumeration of $T_p[i-1, j_{i-1}]$ according to $\prec$ be $T_p[i-1, j_{i-1}] = \{S'_1, \ldots S'_{t_p[i-1, j_{i-1}]}\}$. We analyse separately the case $I_i \leq t_p[i-1,j_i]$ and the case $t_p[i-1,j_i] < I_i$.

  When $I_i \leq t_p[i-1,j_i]$ then $I_{i-1} \leq t_p[i-1, j_{i-1}]$ because $I_{i-1} =I_i$ and $j_{i-1} = j_i$. Also, $S'_{\ell} =S_{\ell}$, for every $\ell \in [1 \tdots t_p[i-1,j_{i-1}]]$. Therefore we have the following equalities
    $S'_{I_{i-1}} = S_{I_i} = S_I \cap [1 \tdots i] = S_I \cap [1 \tdots i-1],$
  where the second equality is the inductive hypothesis, and the third equality holds because $i \not \in S_{I_i}$ when $I_i \leq t_p[i-1,j_i]$.

  When $t_p[i-1,j_i] < I_i \leq t_p[i,j_i]$ then $I_{i-1} \leq t_p[i-1, j_{i-1}]$ because $I_{i-1} =I_i  - t_p[i-1,j_i]$ and $t_p[i-1, j_{i-1}] = t_p[i, j_i] - t_p[i-1,j_i] $ (here we used the definition of $j_{i-1}$). Also, $S'_{\ell} =S_{\ell + t_p[i-1, j_i]} \setminus \{i\} $, for every $\ell \in [1 \tdots t_p[i-1,j_{i-1}]]$. Therefore we have the following equalities:
    $S'_{I_{i-1}} = S_{I_i} \setminus \{i\} = (S_I \cap [1 \tdots i] ) \setminus \{i\} = S_I \cap [1 \tdots i-1],$
  where again the second equality is the inductive hypothesis.
\end{proof}

As a direct corollary, we obtain a new method for enumerating solutions to \Sub.

\begin{corollary}[Enumerating solutions to \Sub\ via dynamic programming]
   \label{corr:better-enumeration}
   Let $A = \{a_1, \ldots , a_n\}$ and $T_{p,k} = \{S \subseteq [n]: \Sigma_A(S) \equiv k \pmod p\}$. For any $c \leq |T_{p,k}|$, it is possible to find~$c$ elements of $T_{p,k}$ in time $\wbo{p + c}$.
\end{corollary}

\begin{proof}
  By Lemma~\ref{lem:dp-table-time} the table $t_p[i,j]$ can be constructed in time $\bo{n^2p}$. Thereafter, by Theorem~\ref{Thm:effEnum} each set $S_I$ for $I \in [1 \tdots t_{p,k}]$ can be computed in time $\bo{n^2}$.
\end{proof}

An alternative method for enumerating solutions was previously known:

\begin{fact}[Enumerating solutions to \Sub\ \cite{BCJ11c}]
  \label{fact:enumeration}
  Let $A = \{a_1, \ldots , a_n\}$ where $a_i = 2^{\bo{n}}$ for all $i$. Let $p = 2^{\bo{n}}$,
  $0 \leq k \leq p-1$, and $T_{p,k} = \{S \subseteq [n]: \Sigma_A(S) \equiv k \pmod p\}$. Then, for any $c \leq |T_{p,k}|$, it is possible to find~$c$ elements of $T_{p,k}$ in time $\wbo{2^{n/2} + c}$.
\end{fact}

In comparison with Fact~\ref{fact:enumeration}, enumerating solutions via dynamic programming is advantageous when~$p < 2^{n/2}$.


\subsection{Statistics about random bins}
\label{Sec:bin}

We describe some statistics about the distribution of the sets $T_{p,k}$ (Definition~\ref{def:ds}) when $b \in (0,1)$ is a constant, $p$ is a random integer in $\bc{2^{bn} \tdots 2^{bn+1}}$, and $k$ is a random integer in $[0 \tdots p-1]$. Therefore, in this section, we stress out that $T_{p,k}$ is a random bin and its cardinality $t_{p,k}$ is a random integer. We first provide an upper bound on the expectation of~$t_{p,k}$.

\begin{lemma}
  \label{lem:size}
  The expected bin size can be upper bounded as $\Ex_{p,k}\bc{t_{p,k}} \leq 2^{(1-b)n}$.
\end{lemma}

 \begin{proof}
The expected size of $T_{p,k}$ is at most $\Ex_{p,k}\bc{t_{p,k}} \leq 2^{(1-b)n}$ since $\set{T_{p,k} : 0 \leq k < p}$ is a partition of~${\cal P}([n])$ with $p \geq 2^{bn}$.
\end{proof}

This result is extended to an upper bound on the second moment of~$t_{p,k}$, under the assumption that the input does not contain too many solution pairs. This bound is needed to analyse the complexity of the variable-time amplitude amplification algorithm.

\begin{lemma}
  \label{lem:binSizeProd}
  Fix any integer $s \geq 0$ and any real $b \in [0,1]$. If there are at most $2^{(2-b)n}$ pairs $(S_1,S_2) \in {\cal P}([n])^2$ such that $\Sigma(S_1) = \Sigma(S_2) + s$, then the expected product of the sizes of two bins at distance $s \bmod p$ from each other is at most $\Ex_{p,k}\bc{t_{p,k}t_{p,(k-s) \bmod p}} \leq \wbo{2^{2(1-b)n}}$.
\end{lemma}

\begin{proof}
  The expectation of $t_{p,k}t_{p,(k-s) \bmod p}$ is equal to the expected number of pairs $(S_1,S_2)$ such that $\Sigma(S_1)$ and $\Sigma(S_2) + s$ are congruent to $k$ modulo $p$, that is $\Ex_{p,k}\bc{t_{p,k} t_{p,(k-s) \bmod p}} = \Ex_{p,k}\bc[\big]{\sum\nolimits_{S_1, S_2} \ind{\Sigma(S_1) \equiv_p \Sigma(S_2) + s \equiv_p k}}$. Since $k$ is uniformly distributed in $[0 \tdots p-1]$, this is equal to $\sum_{S_1, S_2} \Ex_p\bc{\frac{1}{p} \ind{\Sigma(S_1) \equiv_p \Sigma(S_2) + s }} \leq 2^{-bn} \sum_{S_1, S_2} \Pr_p\bc{\Sigma(S_1) \equiv_p \Sigma(S_2) + s}$, using that $p \geq 2^{bn}$. It decomposes as $2^{-bn} \pt[\big]{\sum_{\Sigma(S_1) = s + \Sigma(S_2)} 1 + \sum_{\Sigma(S_1) \neq s + \Sigma(S_2)} \Pr_p\bc{\Sigma(S_1) \equiv_p \Sigma(S_2) + s}}$, where the first inner term is at most $2^{(2-b)n}$ by assumption, and the second term is at most $2^{2n} n 2^{-bn}$ by Fact~\ref{fact:primedivide}. Thus, $\ex{t_{p,k}t_{p,(k-s) \bmod p}} \leq \bo*{2^{-bn} \pt[\big]{2^{(2-b)n} + 2^{2n} n 2^{-bn}}} \leq \wbo{2^{2(1-b) n}}$.
\end{proof}

Finally, we provide a lower bound on the number of \emph{distinct} subset sum values that get hashed to the random bin $T_{p,k}$.

\begin{lemma}
  \label{lem:pz}
  Let $V$ be any subset of the image set $\set{v \in \N : \exists S \subseteq [n], \Sigma(S) = v}$. Let $v_{p,k}$ denote the number of values $v \in V$ such that $v \equiv k \pmod p$. Suppose that $\abs{V} \geq 2^{(1-\ell)n}$ for some $\ell \in [0,1]$. Then,
   \[\left\{\begin{array}{ll}
     \Pr_{p,k}\bc{v_{p,k} \geq 2^{(1-\ell-b)n-2}} = \om{1/n}, & \text{when $\ell \leq 1 - b$}, \\[1mm]
     \Pr_{p,k}\bc{v_{p,k} \geq 1} = \om*{\min\pt*{1/n,2^{(1-\ell-b)n}}},     & \text{when $\ell > 1 - b$.}
   \end{array}\right.\]
\end{lemma}

\begin{proof}
  The expected size of $V_{p,k}$ is at least $\Ex_{p,k}\bc{v_{p,k}} \geq \abs{V}/p \geq \abs{V} 2^{-bn-1}$ since $\{V_{p,k} : 0 \leq k < p\}$ is a partition of~$V$. Similarly to Lemma~\ref{lem:binSizeProd}, the second moment satisfies that $\Ex_{p,k}\bc{v_{p,k}^2} \leq \bo*{2^{-bn} \pt{|V| + \abs{V}^2 n 2^{-bn}}}$ by using Fact~\ref{fact:primedivide}. If $\ell \leq 1-b$ we can further simplify this bound into $\Ex_{p,k}\bc{v_{p,k}^2} \leq \bo*{2^{-bn} |V|}$ since $\abs{V} \geq 2^{(1-\ell)n}$ by assumption. Finally, the result is obtained by applying the Paley–Zygmund inequality $\pr{v_{p,k} \geq \ex{v_{p,k}}/2} \geq \frac{\ex{v_{p,k}}^2} { 4\ex{v_{p,k}^2}}$ and the fact that $\pr{v_{p,k} \geq 1} = \pr{v_{p,k} > 0}$ since $v_{p,k}$ is an integer.
\end{proof}

%% file: subset.tex
As an illustration of the utility of the data structure introduced in Section~\ref{Sec:DynProg}, here we show how it can be used to give quantum and classical algorithms for \emph{worst-case} instances of \Sub\ (Problem~\ref{Pbm:Sub}) based on the representation technique, with running times $\wbo{2^{n/3}}$ and $\wbo{2^{n/2}}$ respectively. These algorithms therefore achieve the same complexity as the best-known algorithms for worst-case complexity based on the meet-in-the-middle principle. Both algorithms use, as a first step, a simple search procedure to handle the case where many solutions exist. Note that the tables constructed by the algorithms do not depend on the target value $m$.

\begin{figure}[htbp]
  \begin{algorithm}[H]
    \caption{Quantum representation technique for \Sub}
    \label{Algo:Quantum-DP-SS}
    \KwIn{Instance of \Sub\ of size $n$ and target $m \leq 2^{4n}$.}
    \KwOut{A subset $S \subseteq [n]$ satisfying $\Sigma(S) = m$ if one exists, otherwise output None.} \vspace{0.3\baselineskip}
    Run the quantum search algorithm (Fact~\ref{fact:search}) over the sets $S \in {\cal P}([n])$, where a set is marked if $\Sigma(S) = m$. Stop it and proceed to step~2 if the running time exceeds~$\wbo{2^{n/3}}$, otherwise output the set it found within the allotted time. \\[1.5mm]
    Choose a random prime $p\in [2^{n/3} \tdots 2^{n/3+1}]$. \\[1mm]
    Construct the table $t_p[i,j]$ for $i=0,\ldots, n$ and $j=0,\ldots, p-1$ (see Section~\ref{Sec:DynProg}). \\[1mm]
    Run quantum search on $T_{p,m \bmod p}$ marking the sets $S \in T_{p,m \bmod p}$ satisfying~$\Sigma(S) = m$.
  \end{algorithm}
\end{figure}

\begin{theorem}[\Sub, quantum]
  \label{thm:Quantum-DP-SS}
  Algorithm~\ref{Algo:Quantum-DP-SS} solves \Sub\ in time $\wbo{2^{n/3}}$ with high probability.
\end{theorem}

\begin{proof}
  If the number of solutions is at least $\abs{\set{S : \Sigma(S) = m}} \geq 2^{n/3}$ then step~1 suffices to solve the problem with high probability since quantum search needs time $\wbo{\sqrt{2^n/2^{n/3}}}=\wbo{2^{n/3}}$. Hence, let us assume that the number of solutions is at most $2^{n/3}$. The table $t_p$ can be constructed in time $\wbo{2^{n/3}}$ according to Lemma~\ref{lem:dp-table-time}. The expected size of $T_{p,m \bmod p}$ can be bounded by $\Ex_p\bc{t_{p,m \bmod p}} = \abs{\set{S : \Sigma(S) = m}} + \sum_{S: \Sigma(S)\neq m}\pr{\Sigma(S) \equiv_p m } = \wbo{2^{2n/3}}$ since $\pr{\Sigma(S) \equiv_p m } \in \bo{n/2^{n/3}}$ by Fact~\ref{fact:primedivide}. Finally, by Markov's inequality, with high probability $t_{p,m \bmod p}$ is no more than a small multiple of this expectation, and quantum search over a bin of size $\wbo{2^{2n/3}}$ takes time $\wbo{2^{n/3}}$ (using the fast oracle of Theorem~\ref{Thm:effEnum}).
\end{proof}

The classical algorithm is similar to the quantum one, except that it constructs a bigger table~$t_p$ to balance the cost of this construction and the cost of the classical search.

\begin{figure}[htbp]
  \begin{algorithm}[H]
    \caption{Classical representation technique for \Sub}
    \label{Algo:Classical-DP-SS}
    \KwIn{Instance of \Sub\ of size $n$ and target $m \leq  2^{{4n}}$.}
    \KwOut{A subset $S \subseteq [n]$ satisfying $\Sigma(S) = m$ if one exists, otherwise output None.} \vspace{0.3\baselineskip}
    Sample $2^{n/2}$ subsets $S \subseteq [n]$ uniformly at random and check if $\Sigma(S) = m$. If no solution is found, proceed to step~2. \\[1mm]
    Choose a random prime $p\in [2^{n/2} \tdots 2^{n/2+1}]$. \\[1mm]
    Construct the table $t_p[i,j]$ for $i=0,\ldots, n$ and $j=0,\ldots, p-1$ (see Section~\ref{Sec:DynProg}). \\[1mm]
    Enumerate the elements of $T_{p,m \bmod p}$ until finding $S \in T_{p,m \bmod p}$ that satisfies~$\Sigma(S) = m$.
  \end{algorithm}
\end{figure}

\begin{theorem}[\Sub, classical]
  \label{thm:Classical-DP-SS}
  Algorithm~\ref{Algo:Classical-DP-SS} solves \Sub\ in time $\wbo{2^{n/2}}$ with high probability.
\end{theorem}

\begin{proof}
  Step~1 does not suffice to solve the problem when the number of solutions is smaller than~$\bo{2^{n/2}}$. Let us suppose that we are in this case. From Lemma~\ref{lem:dp-table-time}, the $(n+1)\times p$ table $t_{p}$ can be constructed in time $\wbo{2^{n/2}}$, after which each query to $T_{p,m \bmod p} = \{S\subseteq[n] : \Sigma(S)\equiv_p m \}$ can be made in time $\bo{n^2}$.  By linearity of expectation, the expected size of $T_{p,m \bmod p}$ can be bounded by $\Ex_p\bc{t_{p,m \bmod p}} = \abs{\set{S : \Sigma(S) = m}} + \sum_{S: \Sigma(S)\neq m}\pr{\Sigma(S) \equiv_p m } = \wbo{2^{n/2}}$ since, by Fact~\ref{fact:primedivide}, $\pr{\Sigma(S) \equiv_p m} \in \bo{n/2^{n/2}}$ for each of the (at most $2^n$) sets $S$ for which $\Sigma(S) \neq m$. By Markov's inequality, with high probability $t_{p,m \bmod p}$ is no more than a small multiple of this expectation and thus can be enumerated in time $\wbo{2^{n/2}}$ using Theorem~\ref{Thm:effEnum}.
\end{proof}

%% file: quantumShiftedSum.tex
In this section we present the two quantum algorithms for solving \SSub\ (Problem~\ref{Pbm:SSub}). The running time of both algorithms -- expressed in Theorems~\ref{Thm:Rep-Dsum} and~\ref{Thm:Mim-Dsum}~-- are functions of the size of a \emph{maximum solution} of the input. This notion plays a central role in our algorithms and is defined next.

\begin{definition}[Maximum solution]
  We say that two disjoint subsets $S_1,S_2 \subseteq [n]$ that form a solution to an instance of \SSub\ are a {\em maximum solution} if the size $\abs{S_1} + \abs{S_2} = \ell n$ is the largest among all such solutions. We call $\ell \in (0,1)$ the maximum solution {\em ratio}.
\end{definition}

By choosing the faster of these two algorithms for each $\ell \in \{1/n , 2/n, \ldots , (n-1)/n\}$ until a solution has been found (or it can be concluded that no solution exists), we obtain an overall quantum algorithm for \SSub\ with the following performance:

\begin{theorem}[\SSub, quantum]
  \label{Thm:Dsum}
  There is a quantum algorithm that, given an instance of \SSub\ with maximum solution ratio $\ell \in (0,1)$, outputs a solution with at least inverse polynomial probability in time $\wbo{2^{\gamma(\ell) n}}$ where
    \begin{numcases}{\gamma(\ell) =}
      (1 + \ell)/4       & if $\ell_1 \leq \ell \leq 3/5$,   \tag*{(Theorem~\ref{Thm:Rep-Dsum})} \\[1mm]
      \ell/2 + 1/10      & if $3/5 < \ell < \ell_2$,   \tag*{(Theorem~\ref{Thm:Rep-Dsum})} \\[1mm]
      (h(\ell) + \ell)/3 & otherwise    \tag*{(Theorem~\ref{Thm:Mim-Dsum})}
    \end{numcases}
  and $\ell_1 \approx 0.190$ and $\ell_2\approx 0.809$ are solutions to the equations $(h(\ell) + \ell)/3 = (1+\ell)/4$ and $(h(\ell) + \ell)/3 = \ell/2 + 1/10$ respectively. In particular, the worst-case complexity is~$\bo{2^{0.504 n}}$.
\end{theorem}

Since a potential solution can be verified in polynomial time in $n$, in what follows we describe our algorithms on yes instances with maximum solution ratio $\ell$. As presented, the algorithms find a solution with inverse polynomial probability in $n$, which can be amplified to constant probability in polynomial time.
Recall the classical algorithm of ~\cite{MNPW19c} for \Equ\ is based on the concept of a \emph{minimum} solution (Definition~\ref{def:min-sol}), rather than a \emph{maximum} solution. In Appendix~\ref{app:quantum-equal}, we present an analogous quantum algorithm for \Equ\ whose complexity is expressed in terms of the minimum solution ratio $\lmin$ (we do not know how to extend this result to \SSub). While this does not change the algorithmic complexity in the worst case, for a given instance of \Equ\ the quantity~$\lmin$ may be smaller than $\ell$.


\subsection{Representation technique algorithm}
\label{Sec:repEsum}

Our representation-technique-based algorithm is given in Algorithm~\ref{Algo:Rep-Dsum}, and uses the dynamic programming table of Section~\ref{Sec:DynProg}. Before constructing that table, we first check whether the input contains many solution pairs (in which case a simple quantum search is sufficient). Depending on the value of the maximum solution ratio $\ell$, we may also need to apply variable-time amplitude amplification (Fact~\ref{fact:vtaa}) on top of quantum pair finding (Theorem~\ref{Thm:PairFinding}).

\begin{figure}[htbp]
  \begin{algorithm}[H]
    \caption{Quantum representation technique for \SSub}
    \label{Algo:Rep-Dsum}
    \KwIn{Instance $(a,s)$ of \SSub\ with $\sum_{i=1}^n a_i < 2^{4n}$ and maximum solution ratio $\ell$.}
    \KwOut{Two subsets $S_1,S_2 \subseteq [n]$.} \vspace{0.3\baselineskip}
    Set $b = (1 + \ell)/4$ if $\ell \leq 3/5$ and $b = 2/5$ if $\ell > 3/5$. \\
    Run the quantum search algorithm (Fact~\ref{fact:search}) over the set of pairs $(S_1,S_2) \in {\cal P}([n])^2$, where a pair is marked if $\Sigma(S_1) = \Sigma(S_2) + s$ and $S_1 \neq S_2$. Stop it and proceed to step~3 if the running time exceeds~$\wbo{2^{bn/2}}$, otherwise output the pair it found within the allotted time. \\[1.5mm]
    If $\ell > 3/5$ then run variable-time amplitude amplification (Fact~\ref{fact:vtaa}) on steps~4--6, otherwise run them once: \\
    \Indp Choose a random prime $p\in [2^{bn} \tdots 2^{bn + 1}]$ and a random integer $k \in [0 \tdots p-1]$. \\
    Construct the table $t_p[i,j]$ for $i=0,\ldots, n$ and $j=0,\ldots, p-1$ (see Section~\ref{Sec:DynProg}). \\
    Run the quantum pair finding algorithm (Theorem~\ref{Thm:PairFinding}) to find if there exist two sets $S_1 \in T_{p,k}$ and $S_2 \in T_{p,(k-s) \bmod p}$ such that $\Sigma(S_1) = \Sigma(S_2)+s$ and $S_1 \neq S_2$. If so, output the pair $(S_1, S_2)$ it found.
  \end{algorithm}
\end{figure}

Note that `run variable-time amplitude amplification on steps~4--6' means that one should apply the procedure implicit in Fact~\ref{fact:vtaa} to the algorithm $\mathcal{A}$ defined by the following process (i) Create a uniform superposition over all primes $p\in [2^{bn} \tdots 2^{bn + 1}]$ and, for each $p$, all $k \in [0 \tdots p-1]$. (ii) For each $p$, coherently construct the table $t_p$. (iii) Run quantum pair finding coherently on each pair of sets $T_{p,k}, T_{p,(k-s)\bmod p}$, marking the $(p,k)$ tuple if a pair is found.

The analysis of the above algorithm relies on the random bin statistics presented in Section~\ref{Sec:bin}. We first define the \emph{collision values set} which contains the values of all the possible solution pairs.

\begin{definition}[\sc Collision Values set]
 Given an instance $(a,s)$ to the $\SSub$ problem, the \emph{collision values set} is the set $V = \set{v \in \N : \exists S_1 \neq S_2, v = \Sigma(S_1) = \Sigma(S_2) + s}$.
\end{definition}

We show that the collision values set $V$ is of size at least $2^{(1-\ell)n}$ when the maximum solution ratio is~$\ell$. Thus, by Lemma~\ref{lem:pz}, we can lower bound the number of values in $V$ that get hashed to a random bin $T_{p,k}$.

\begin{lemma}
  \label{lem:maxRatioEsum}
  If the maximum solution ratio is $\ell$ then $|V| \geq 2^{(1-\ell)n}$.
\end{lemma}

\begin{proof}
  Let $S_1, S_2\subseteq \{1, \ldots, n\}$ be a maximum solution of size $\abs{S_1}+\abs{S_2} = \ell n$. Then for any $S \subseteq  [n] \setminus \pt{S_1 \cup S_2}$ the sets $S_1 \cup S$ and $S_2 \cup S$ form a solution, and for $S \neq S'$, the values $\Sigma(S_1 \cup S) $ and $\Sigma(S_1 \cup S') $ must be distinct. If this were not the case then $S_1 \cup \pt{S \setminus S'}$ and $S_2 \cup \pt{S' \setminus S}$ would form a disjoint solution of size larger than $\ell$.
\end{proof}

We finally analyse Algorithm~\ref{Algo:Rep-Dsum} in the next theorem.

\begin{theorem}[\SSub, representation]
  \label{Thm:Rep-Dsum}
  Given an instance of \SSub\ with $\sum_{i=1}^n a_i < 2^{4n}$  and maximum solution ratio $\ell \in (0,1)$, \emph{Algorithm~\ref{Algo:Rep-Dsum}} finds a solution with inverse polynomial probability in time $\wbo{2^{(1+\ell) n/4}}$ if $\ell \leq 3/5$, and $\wbo{2^{(\ell/2+1/10) n}}$ if $\ell > 3/5$.
\end{theorem}

\begin{proof}
  Step~2 of Algorithm~\ref{Algo:Rep-Dsum} handles the case where the total number of solution pairs exceeds~$2^{(2-b)n}$. In this situation, the quantum search algorithm can find a solution pair in time $\wbo{\sqrt{2^{2n}/2^{(2-b)n}}} = \wbo{2^{bn/2}}$, which is smaller than the complexity stated in Theorem~\ref{Thm:Rep-Dsum}.

  {\it Analysis when $\ell \leq 3/5$.}
  In this case the algorithm executes steps~4--6 only once. From Lemma~\ref{lem:dp-table-time}, the table $t_p$ can be constructed in time $\wbo{2^{bn}}$, after which each query to the elements of $T_{p,k}$ can be performed in time $\bo{n^2}$ (Theorem~\ref{Thm:effEnum}). By Lemma~\ref{lem:pz}, the number of disjoint solution pairs contained in~$T_{p,k} \times T_{p,(k-s) \bmod p}$ is at least $v_{p,k} \geq 2^{(1-\ell-b)n-2}$ with probability $\om{1/n}$. By Lemma~\ref{lem:size} and Markov's inequality, the sizes of $T_{p,k}$ and $T_{p,(k-s) \bmod p}$ are at most $t_{p,k}, t_{p,(k-s) \bmod p} \leq n^2 2^{(1-b)n}$ with probability at least $1-1/n^2$. Thus, with probability $\om{1/n}$ we can assume that both of these events occur. If this is the case, then the time to execute step~6 of the algorithm is $\wbo[\Big]{\pt[\big]{t_{p,k}t_{p,(k-s) \bmod p}/v_{p,k}}^{1/3}} = \wbo{2^{(1+\ell-b)n/3}}$ since the first complexity given in Theorem~\ref{Thm:PairFinding} is the largest one for our choice of parameters. This is at most~$\wbo{2^{(1+\ell)n/4}}$ when~$b = (1+\ell)/4$.

  {\it Analysis when $\ell > 3/5$.}
  We assume that the total number of solution pairs is at most $2^{(2-b)n}$ (otherwise we would have found a collision at step~2 with high probability). Given $p$ and $k$, the base algorithm (steps~4--6) succeeds if there is a solution in $T_{p,k} \times T_{p,(k-s) \bmod p}$, i.e.~$v_{p,k} \geq 1$. Therefore by Lemmas~\ref{lem:maxRatioEsum} and \ref{lem:pz}, we have for its success probability $\rho = \om{\min\pt{1/n,2^{(1-\ell-b)n}}}$. We claim that $\ex{\tau^2} = \wbo{2^{2bn}}$ where $\tau$ is the stopping time of the base algorithm. Constructing the table $t_p$ takes time $\wbo{p}$, and by summing the two complexities given in Theorem~\ref{Thm:PairFinding} the quantum pair finding algorithm takes time at most $\wbo[\big]{\pt{t_{p,k}t_{p,(k-s)\bmod p}}^{1/3} + \sqrt{\max\pt{t_{p,k}, t_{p,(k-s)\bmod p}}}}$. Therefore we have
    \begin{align*}
      \ex{\tau^2}
       & = \wbo[\bigg]{\Ex_{p,k}\bc[\bigg]{\pt[\Big]{p + \pt[\big]{t_{p,k}t_{p,(k-s)\bmod p}}^{1/3} + \sqrt{\max\pt{t_{p,k}, t_{p,(k-s)\bmod p}}}}^2} } \\
       & \leq \wbo[\Big]{\Ex_{p,k}\bc[\big]{p^2}
          + \Ex_{p,k}\bc[\Big]{t^{2/3}_{p,k}t^{2/3}_{p,(k-s)\bmod p}}
          + \Ex_{p,k}\bc[\big]{t_{p,k}}} \\
       & \leq \wbo[\Big]{\Ex_{p,k}\bc[\big]{p^2}
          + \Ex_{p,k}\bc[\Big]{t_{p,k}t_{p,(k-s)\bmod p}}^{2/3}
          + \Ex_{p,k}\bc[\big]{t_{p,k}}} \\
       & \leq \wbo{2^{2bn}} + \wbo{2^{4(1-b)n/3}} + 2^{(1-b)n},
    \end{align*}
  where the second inequality uses that the moment function is non-decreasing and the last inequality uses Lemmas~\ref{lem:size} and~\ref{lem:binSizeProd}. Since $b = 2/5$ we obtain that $\ex{\tau^2} \leq \wbo{2^{2bn}}$. Finally, by Fact~\ref{fact:vtaa}, the overall time of steps~3--6 is $\wbo{\sqrt{\ex{\tau^2}/\rho}} = \wbo{2^{bn}/2^{(1-\ell-b)n/2}} = \wbo{ 2^{  (\frac{\ell}{2}  + \frac{1}{10} ) n } }$.
\end{proof}


\subsection{Meet-in-the-middle algorithm}
\label{Sec:Mim-Dsum}

Our second algorithm uses the standard meet-in-the-middle technique combined with quantum search to solve the \SSub\ problem. We first state a lemma that if we randomly partition the input into two sets of relative sizes 1:2, then with at least inverse polynomial probability  a maximum solution will be distributed in the same proportion between the two sets.

\begin{lemma}
  \label{lem:poly_prob}
  Let $S_1, S_2$ be a maximum solution of ratio $\ell$. Then with at least inverse polynomial probability the random partition $ X_1\cup X_2$ satisfies $\abs{(S_1\cup S_2)\cap X_1} = \ell n/3$, $\abs{(S_1\cup S_2)\cap X_2} = 2\ell n/3$.
\end{lemma}

\begin{proof}
  There are $\binom{n}{n/3}$ ways to partition $[n]$ into two subsets $X_1$ and $X_2$ of respective sizes $n/3$ and $2n/3$. Of these, there are $\binom{\ell n}{\ell n/3}\cdot\binom{n-\ell n}{\frac{n-\ell n}{3}}$ partitions such that $\abs{(S_1\cup S_2)\cap X_1} =\ell n/3, \abs{(S_1\cup S_2)\cap X_2} = 2\ell n/3$. The probability that $\abs{(S_1\cup S_2)\cap X_1} =\ell n/3, \abs{(S_1\cup S_2)\cap X_2} = 2\ell n/3$ is thus $\frac{\binom{\ell n}{\ell n/3}\cdot\binom{n-\ell n}{\frac{n-\ell n}{3}} }{\binom{n}{n/3}}$. Fact~\ref{fact:entropy} gives that this quantity is at least $\om[\big]{n^{-1/2}}$.
\end{proof}

We use the above result in the design of Algorithm~\ref{Algo:Mim-Dsum}, which is analysed in the next theorem. We observe that the obtained time complexity is at most $\wbo{3^{n/3}}$ and is maximized at $\ell = 2/3$.

\begin{figure}[htbp]
  \begin{algorithm}[H] \label{Algo:Mim-Dsum}
    \caption{Quantum meet-in-the-middle technique for \SSub}
    \KwIn{Instance $(a,s)$ of \SSub\ with maximum solution ratio $\ell$.}
    \KwOut{Two subsets $S_1,S_2 \subseteq [n]$.} \vspace{0.3\baselineskip}
    Randomly split $[n]$ into disjoint subsets $X_1\cup X_2$ such that $\abs{X_1}=n/3, \abs{X_2}=2n/3$. \\[1.5mm]
    Classically compute and sort~$V_1 = \{ \Sigma(S_{11}) - \Sigma(S_{21}) : S_{11},S_{21}\subseteq X_1 {\rm ~and~} \allowbreak S_{11}\cap S_{21}=\emptyset \allowbreak {\rm ~and~} \allowbreak \abs{S_{11}}+\abs{S_{21}} = \ell n/3\}$. \\[1.5mm]
    Apply quantum search (Fact~\ref{fact:search}) over the set~$V_2 = \set[\big]{ \Sigma(S_{12}) - \Sigma(S_{22}) :  S_{12},S_{22}\subseteq X_2  \allowbreak  {\rm and~} S_{12}\cap S_{22}=\emptyset  \allowbreak {\rm ~and~} \allowbreak \abs{S_{12}}+\abs{S_{22}} = 2\ell n/3}$, where an element $v_2\in V_2$ is marked if there exists $v_1\in V_1$ such that $v_1 + v_2 = s$. For a marked $v_2$, output $S_1 = S_{11} \cup S_{12}$ and
    $S_2 = S_{21} \cup S_{22}$.
  \end{algorithm}
\end{figure}

\begin{theorem}[\SSub, meet-in-the-middle]
  \label{Thm:Mim-Dsum}
  Given an instance of \SSub\ with maximum solution ratio $\ell \in (0,1)$, \emph{Algorithm~\ref{Algo:Mim-Dsum}} finds a solution with at least inverse polynomial probability in time $\wbo{2^{n (h(\ell) + \ell)/3}}$.
\end{theorem}

\begin{proof}
  There are $\binom{n/3}{\ell n/3}2^{\ell n/3}$ different ways to select two sets $S_{11},S_{21}\subseteq X_1$ such that $S_{11}\cap S_{21}=\emptyset$, $\abs{S_{11}}+\abs{S_{21}} = \ell n/3$.  Computing and sorting $V_1$ thus take time $\wbo{\binom{n/3}{\ell n/3}2^{\ell n/3}}$.
  In the next step of the algorithm, quantum search is performed over all $\binom{2n/3}{2\ell n/3}2^{2\ell n/3}$ sets $S_{12}, S_{22}\subseteq X_2$ such that $S_{12}\cap S_{22} = \emptyset$, $\abs{S_{12}}+\abs{S_{22}} = 2\ell n/3$.  We mark an element $v_2\in V_2$ if there exists $v_1 \in V_1$ such that $v_1 + v_2 = s$. Since $V_1$ is sorted this check can be done in time $\polylog\pt{\abs{V_1}}$. The total time required is therefore
    $\wbo*{\binom{n/3}{\ell n/3}2^{\ell n / 3} + \sqrt{\binom{2n/3}{2\ell n/3}2^{2\ell n / 3}}}
      = \wbo*{2^{\ell n/3 }\pt*{\binom{n/3}{\ell n/3} + \sqrt{\binom{2n/3}{2\ell n/3}}}}
      = \wbo{2^{\frac{n}{3}\pt{h(\ell) + \ell}}}$.
  By Lemma~\ref{lem:poly_prob}, when the instance has a maximum solution of size $\ell n$, the set
  $V_2$ has a marked element with at least inverse polynomial probability, and in that case a solution is found.
\end{proof}

%% file: pigeonhole.tex
\subsection{\PEqu}
\label{sec:non-mod}

We give classical and quantum algorithms for \PEqu\ (Problem~\ref{Pbm:PEqu}), based on dynamic programming and which run in time $\wbo{2^{n/2}}$ and $\wbo{2^{2n/5}}$, respectively. In contrast with our quantum algorithm for \SSub\ that made use of a random prime modulus, in the case of \PEqu\ we can deterministically choose a modulus~$p$, and the pigeonhole principle guarantees a collision in at least one bin.

\begin{lemma}
  \label{lem:pigeonhole-collision}
  There is a classical deterministic algorithm such that, given an instance of \PEqu\ and a modulus $p$ that divides $2^n$, it finds in time $\wbo{p}$ an integer~$k$ such that there exist two distinct subsets $S_1, S_2$ with $\Sigma(S_1) \equiv \Sigma(S_2) \equiv k \pmod p$.
\end{lemma}

\begin{proof}
  Denote by $\br{0}, \br{1}, \dots \br{p-1}$ the congruence classes modulo $p$. Each of these classes contains exactly $2^n/p$ numbers between $0$ and $2^{n}-2$, except the last class $\br{p-1}$ that has only $2^n/p-1$ numbers. Since all $2^n$ subsets $S \subseteq [n]$ have a sum $\Sigma(S)$ between $0$ and $2^n-2$ there are two possible cases:
  \begin{itemize}
  	\item either there is some class $\br{k}$ such that $\Sigma(S) \in \br{k}$ for strictly more than $2^{n}/p$ subsets $S$,
    \item or there are $2^{n}/p$ subsets $S$ such that $\Sigma(S) \in \br{p-1}$.
  \end{itemize}

  Denote by $\br{k}$ a class that verifies one of these two points. By definition, there are \textit{strictly more} subsets $S$ such that $\Sigma(S) \in \br{k}$ than the number of elements between $0$ and $2^n-2$ that belong to~$\br{k}$. However, for all $S \subseteq [n]$, we have $\Sigma(S) \leq 2^n-2$. Thus, there must be two subsets $S_1 \neq S_2$ such that $\Sigma(S_1), \Sigma(S_2) \in \br{k}$ and $\Sigma(S_1) = \Sigma(S_2)$.

  From Lemma~\ref{lem:dp-table-time}, the table $t_p[i,j] = \abs*{\set*{S \subseteq \{1,\dots,i\} : \Sigma(S) \equiv j \pmod p}}$ can be constructed in time $\wbo{p}$. From the table, we can read off a value $k$ that satisfies the above condition.
\end{proof}

\begin{theorem}[\PEqu, classical]
  \label{thm:phsAlgo}
  There is a classical deterministic algorithm for the \PEqu\ problem that runs in time $\wbo{2^{n/2}}$.
\end{theorem}

\begin{proof}
  Choose $p=2^{n/2}$. By Lemma~\ref{lem:pigeonhole-collision}, in time $\wbo{2^{n/2}}$ we can find $k$ such that there exist $S_1 \neq S_2$ satisfying $\Sigma(S_1) \equiv \Sigma(S_2) \equiv k \pmod {2^{n/2}}$. Once we know a bin that contains a collision, by Corollary~\ref{corr:better-enumeration}, we can enumerate in time $\wbo{2^{n/2}}$ a sufficient number of subsets in that bin to locate a collision.
\end{proof}

\begin{theorem}[\PEqu, quantum]
  \label{thm:phsq}
  There is a quantum algorithm for the \PEqu\ problem that runs in time $\wbo{2^{2n/5}}$.
\end{theorem}

\begin{proof}
  We set $p = 2^{2n/5}$ and, by Lemma~\ref{lem:pigeonhole-collision}, in time $\wbo{2^{2n/5}}$ we can identify $k$ such that there exist $S_1\neq S_2$ satisfying $\Sigma(S_1) \equiv \Sigma(S_2) \equiv k \pmod {2^{2n/5}}$. By Theorem~\ref{Thm:effEnum}, each query to $T_{p,k} = \{S\subseteq[n] : \Sigma(S)\equiv k \pmod p\}$ can be made in time $\bo{n^2}$. We use Ambainis' element distinctness algorithm~\cite{Amb07j} on these elements to find a collision. We do not want to run it on an unnecessarily large set. Therefore, if $t_{p,k} > 2^{3n/5 +1}$ then we run it only on the first $2^{3n/5 +1}$ elements of $T_{p,k}$, according to the ordering defined by $\prec$. A collision is then found in time $\wbo[\big]{\pt{2^{3n/5}}^{2/3}} = \wbo{2^{2n/5}}$. The overall running time of the algorithm is thus $\wbo{2^{2n/5}}$.
\end{proof}


\subsection{\PMEqu}
\label{Sec:PMAlgo}

 Here we give a classical representation-technique-based algorithm for \PMEqu\ (Problem~\ref{Pbm:PMEqu}). Our approach consists of defining the bins based on the \emph{quotient} in the division by $p$ instead of the remainder, that is $T_{p,k} = \set*{S\subseteq \{1, \ldots , n\} : \floor{\pt{\sum_{i \in S} a_i \bmod q} / p} = k}$. We show that computing the cardinality of~$T_{p,k}$ and enumerating its elements can be done with the help of yet another table for the bins $T'_{p,k} = \set*{S\subseteq \{1, \ldots , n\} : \pt{\sum_{i \in S} \floor{a_i / p}} \equiv k \pmod {\floor{q/p}}}$. This last table can be constructed with a similar dynamic programming technique as before. We do not know how to extend this approach to the more general \Equ\ or \SSub\ problems due to the lack of good statistics on how a random bin $T_{p,k}$ behaves in this case. We also have no quantum speed-up for \PMEqu\ due to a bottleneck when going from $T'_{p,k}$ to $T_{p,k}$ that we cannot seemingly reduce with quantum techniques.

\begin{theorem}[\PMEqu, classical]
  \label{thm:phPMSub}
  There is a classical deterministic algorithm for the \PMEqu\ problem that runs in time~$\wbo{2^{n/2}}$.
\end{theorem}

\begin{proof}
  Without loss of generality we suppose that for every input $a_i$, the inequality $a_i < q$ holds, where $q \leq 2^n - 1$ is the modulus in the input. For such a modulus there exists a unique couple $(q_1,q_2)$ with $0 \leq q_1, q_2 < 2^{n/2}$ such that $q = q_1 2^{n/2} + q_2$. We define the one-dimensional array $B[j]$ for $0\leq j \leq q_1 -1$ by
    \[B[j] =  \set*{S \subseteq \{1,\dots,n\} :  \floor*{\frac{ \Sigma(S) \bmod q}{2^{n/2}}} =j }.\]
  We denote the cardinality of $B[j]$ by $b[j]$, and we define $\beta[j] = \abs*{\set*{ 0 \leq k \leq q-1 : \floor*{\frac{k}{2^{n/2}}} =j}}.$ Observe that $\beta[j] \leq 2^{n/2}$, for all $j$, and that
    \[\sum_{j=0}^{q_1 - 1} \beta_j = q < 2^n =  \sum_{j=0}^{q_1 - 1} b[j].\]
  Therefore there exists $j$ such that $\beta[j] < b[j]$, and we will call such an index {\em marked}. The algorithm will identify a marked $j$, and then will find $\beta[j] +1$ different sets in $B[j]$. We will show that this can be done in time $\wbo{2^{n/2}}$, and by the pigeonhole principle there are two sets $S_1, S_2  \subseteq \{1,\dots,n\}$ among them such that $ \Sigma(S_1) \equiv \Sigma(S_2) \pmod q$.

  Computing directly the values in the array $b$ is not easy, we will do that with the help of another one-dimensional array $C$. For $1 \leq k \leq n,$ we set $a'_k = \floor*{\frac{ a_k }{2^{n/2}}}$ and $A' = \{a'_1, \ldots, a'_n\}$.
  Then we define
     \[C[j] =  \left\{S \subseteq \{1,\dots,n\} :  \Sigma_{A'}(S) \equiv j \pmod {q_1} \right\},\]
  for $0\leq j \leq q_1 - 1$, and we set $c[j] = |C[j]|$. By Lemma~\ref{lem:dp-table-time} the full array $c$ can be computed in time $\wbo{2^{n/2}}$, and by Theorem~\ref{Thm:effEnum}, for every $0\leq j \leq q_1 -1$, the entry $C[j]$ can be enumerated in $\bo{n^2}$ time per set.

  The arrays $B$ and $C$ are of course not identical, but the following lemma shows that any set in $B[j]$ must be contained in $C$ at an index close to $j$.  For any $0 \leq i,j \leq q_1 - 1$, we define
    \[B[i \tdots j] =
     \begin{cases}
       \cup_{k=i}^j B[k] & \text{if $i \leq j$} \\[1mm]
       \cup_{k=i}^{q_1-1} B[k] \bigcup \cup_{k=0}^j B[k] & \text{otherwise},
     \end{cases}\]
  and $C[i \tdots j]$ is defined analogously. We denote their respective cardinalities as $b[i \tdots j] = |B[i \tdots j]|$ and $c[i \tdots j] = |C[i \tdots j]|$. Finally, for $0 \leq i \leq j \leq q_1$, we set $\beta[i \tdots j] = \sum_{k=i}^j \beta[k]$.

  \begin{lemma}
  \label{lem:BsubC}
    For every $0 \leq j \leq q_1-1$, the inclusion
      $B[j] \subseteq C[(j -n+1) \bmod q_1 \tdots (j+n-1) \bmod q_1]$
    holds.
  \end{lemma}

  \begin{proof}
    Let $S \in B[j]$. Then by definition $\Sigma(S) \bmod q = j \cdot 2^{n/2} + j'$, for some $0 \leq j' \leq 2^{n/2} -1$. Consequently $\Sigma(S) = k_1 q + j \cdot 2^{n/2} + j'$, where $0 \leq k_1 \leq n-1$ because $a_i < q$ for every $i$, and therefore $\Sigma(S) < nq$. This implies that
      \begin{eqnarray*}
         \floor*{\frac{ \Sigma(S)}{2^{n/2}}}  & =  & j +  \floor*{\frac{ k_1 q + j'}{2^{n/2}}} \\
        & = & j +   \floor*{\frac{ k_1(q_1 2^{n/2} + q_2 ) + j'}{2^{n/2}}} \\
        & = & j + k_1 q_1 + \floor*{\frac{k_1 q_2 + j'}{2^{n/2}}} \\
        & = &  j + k_1 q_1  + k_2,
      \end{eqnarray*}
    where $0 \leq k_2 \leq n-1.$ Therefore
      \[\Sigma_{A'}(S) =  j + k_1 q_1  + k_2 - k_3,\]
    where $0 \leq k_3 \leq n-1$ since the set $S$ contains at most $n$ elements. We can thus conclude that
      \[\Sigma_{A'}(S) \equiv j + k_4 \pmod {q_1},\]
    where for $ k_4 = k_2 - k_3 $ we have $-n+1 \leq k_4 \leq n-1$.
  \end{proof}

  \begin{corollary}
    \label{cor:incl}
    Let $0 \leq i,j \leq q_1 -1$ with $2n \leq j - i \leq (q_1 - 1)/2$. Then
      $C[i+n-1 \tdots j-n+1] \subseteq B[i \tdots j] \subseteq  C[(i-n+1) \bmod q_1 \tdots (j+n-1) \bmod q_1]$.
  \end{corollary}

  \begin{proof}
    The second inclusion follows immediately from Lemma~\ref{lem:BsubC}. The lemma also implies
      $B[j+1 \bmod q_1 \tdots i-1 \bmod q_1] \subseteq  C[(j-n+2)  \tdots (i+n-2) ]$.
    Taking the complement of the set on each side gives the first inclusion.
  \end{proof}

  \begin{corollary}
    \label{cor:CsubB}
    Let $0 \leq i \leq q_1 -1$. Then there exists $k \in \{i - n, i + n\}$ such that
     $C[i] \subseteq  B[(k- 2n +1)  \tdots (k+2n-1) ]$.
  \end{corollary}

  \begin{proof}
    We have either $i \geq 3n-1$ or $i \leq q_1 -3n$. If only the first case is true choose $k=i-n$, if only the second case is true choose $k = i+n$, if both cases are true choose arbitrarily. Obviously $C[i] \subseteq C[(k-n) \tdots (k+n)]$, therefore Corollary~\ref{cor:incl} implies the statement.
  \end{proof}

   \begin{corollary}
    \label{cor:enumerate}
    Let $0 \leq i \leq q_1 -1$.  In time $\wbo{2^{n/2}}$ we can either enumerate $C[i]$ or we can find a marked index.
  \end{corollary}

  \begin{proof}
    If $c[i] \leq (4n-1) 2^{n/2}$ then by Theorem~\ref{Thm:effEnum} we can enumerate $C[i]$. Otherwise, by Corollary~\ref{cor:CsubB}, we test each $S\in C[i]$ to see which set $B[j]$, for $k-2n+1 \leq j \leq k+2n -1,$ it belongs to until an index~$j$ satisfying $b[j]>2^{n/2}$ is identified.
  \end{proof}

  We now describe the procedure to find a marked index. It is essentially a dichotomic search over shorter and shorter intervals $[i \tdots j] = \{i, \ldots, j\},$ with the invariant property $\beta[i \tdots j] < b[i \tdots j],$ and where in every step we halve the size of $j-i$. Initially, we set $i = 0$ and $j = q_1 -1$, and we stop the process when $2n \leq j-i < 4n.$ Clearly the number of iterations is less than $n/2$. We now describe one iteration. Let us suppose that our current interval is $[i \tdots j]$, and let $m=(j-i)/2$, rounding it arbitrarily, if necessary. We will compute $b[i \tdots m]$ and $b[m+1 \tdots j]$ and keep one of the two intervals for which the invariant property holds.

  We claim that for some fixed indices $i,j$, in time $\wbo{2^{n/2}}$ we can either compute~$b[i \tdots j]$  or we find a marked index. From Corollary~\ref{cor:incl} it follows that
    \begin{align*}
      b[i \tdots j] & = c[i+n-1 \tdots j-n+1] \\
                    & + \abs[\Big]{B[i \tdots j] \cap \pt[\big]{C[(i-n+1) \bmod q_1 \tdots i+n-2] \cup C[(j-n+2) \tdots (j+n-1) \bmod q_1]}}.
    \end{align*}
  The first term $c[i+n-1 \tdots j-n+1]$ is computed in time $\wbo{q_1}$ by adding the corresponding entries in the array $c$. For the second term, since there at most $4n$ entries of~$C$ involved in it, we can either enumerate all the elements they contain in time $\wbo{2^{n/2}}$ by Corollary~\ref{cor:enumerate} and check for each element if it belongs to $B[i \tdots j]$ (hence we can compute $b[i \tdots j]$), or we can find a marked index.

  Unless we already found a marked index  during the process, the dichotomic search stops with less than $4n$ candidate indices out of which at least one is marked. Therefore the last thing to show is that given a marked index $j$, how do we find a solution in time $\wbo{2^{n/2}}$. From Lemma~\ref{lem:BsubC} and Corollary~\ref{cor:incl} we know that
    \[B[j] \subseteq C[(j -n+1) \bmod q_1 \tdots (j+n-1) \bmod q_1] \subseteq  B[j - 2n + 2 \bmod q_1 \tdots j + 2n + 2 \bmod q_1].\]
  We start enumerating ${\cal C} = C[(j -n+1) \bmod q_1 \tdots (j+n-1) \bmod q_1]$ until one of the following two things happens. If $|{\cal C}| \leq (4n-3) 2^{n/2}$ then we fully enumerate ${\cal C}$ and therefore can also fully enumerate $B[j]$, and find a solution there. Otherwise, we stop after having enumerated $(4n-3) 2^{n/2} + 1$ elements of ${\cal C}$, and for each of them we determine the index where they belong in $B[j - 2n + 2 \bmod q_1 \tdots j + 2n + 2 \bmod q_1].$ There will be an index where we have found more than $2^{n/2}$ subsets and therefore also a solution.
\end{proof}

\begin{remark}
  As an anonymous referee has pointed out, \PMEqu\ can also be solved in the same running time as above based on a meet-in-the-middle approach: (i) Create a  list $L_1$ of all subsets $S\subseteq \{1, \ldots, n/2 \} $ sorted in non-decreasing order of $\Sigma(S) \bmod q$, and another similar list $L_2$ for subsets $S\subseteq \{n/2 + 1, \ldots, n\}$. (ii) By performing $\bo{2^{n/2}}$ binary searches through $L_1$ (one for each element of $L_2$), one can compute $\abs{\{S\subseteq [n] : \Sigma(S) \bmod q \le x\}}$ (for any $x$), and hence compute $\abs{\{S\subseteq [n] : \Sigma(S)\bmod q \in [l,r]\}}$. (iii) Starting from the interval $[0, q-1]$, for each interval $[l,r]$ and $m= (l+r)/2$ recursively identify whether $[l,m]$ or $[m,r]$ contains more subsets than pigeonholes. The problem is solved in $\poly (n)$ recursions.
\end{remark}

%% file: acknowledgements.tex
This work has been supported by the European Union’s H2020 Programme under grant agreement number ERC-669891.
Research at CQT is funded by the National Research Foundation, the Prime Minister’s Office, and the Ministry of Education, Singapore under the Research Centres of Excellence programme’s research grant R-710-000-012-135.
JA thanks Shengyu Zhang for helpful discussions during the course of this work.

%% file: classicalShiftedSum.tex
Here we adapt the classical $\bo{2^{0.773 n}}$ algorithm for \Equ\ of Mucha et al.~\cite{MNPW19c} to apply to the more general \SSub\ problem (Problem~\ref{Pbm:SSub}), with the same worst-case running time.

\begin{figure}[htbp]
  \begin{algorithm}[H] \label{Algo:classical-Mim-Dsum}
    \caption{Classical meet-in-the-middle technique for \SSub}
    \KwIn{Instance $(a,s)$ of \SSub\ with maximum solution ratio $\ell$.}
    \KwOut{Two subsets $S_1,S_2 \subseteq [n]$.} \vspace{0.3\baselineskip}
    Randomly split $[n]$ into disjoint subsets $X_1\cup X_2$ such that $\abs{X_1}=\abs{X_2}=n/2$. \\[1.5mm]
    Classically compute and sort~$V_1 = \{ \Sigma(S_{11}) - \Sigma(S_{21}) : S_{11},S_{21}\subseteq X_1 {\rm ~and~} S_{11}\cap S_{21}=\emptyset \allowbreak {\rm ~and~} \allowbreak \abs{S_{11}}+\abs{S_{21}} = \ell n/2\}$. \\[1.5mm]
    Classically compute~$V_2 = \{ \Sigma(S_{12}) - \Sigma(S_{22}) : S_{12},S_{22}\subseteq X_2 {\rm ~and~} \allowbreak  S_{12}\cap S_{22}=\emptyset  \allowbreak {\rm ~and~} \allowbreak \abs{S_{12}}+\abs{S_{22}} = \ell n/2\}$ \\[1.5mm]
    For each $v_2\in V_2$, binary search for $v_1\in V_1$ such that $v_1 + v_2 = s$. If such a $v_2$ is found, output $S_1 = S_{11} \cup S_{12}$ and
    $S_2 = S_{21} \cup S_{22}$, where $S_{11}, S_{21}\subseteq X_1$ are such that $v_1 = \Sigma(S_{11}) - \Sigma(S_{21})$ and $S_{21},S_{22}$ are such that $v_2 = \Sigma(S_{21})-\Sigma(S_{22})$.
  \end{algorithm}
\end{figure}

\begin{theorem}[\SSub, classical meet-in-the-middle]
  \label{Thm:classical-Mim-Dsum}
  Given an instance of \SSub\ with maximum solution ratio $\ell \in (0,1)$, \emph{Algorithm~\ref{Algo:classical-Mim-Dsum}} finds a solution with at least inverse polynomial probability in time $\wbo{2^{n (h(\ell) + \ell)/2}}$.
\end{theorem}

\begin{proof}
  By the the same proof technique as Lemma~\ref{lem:poly_prob}, it follows that with at least inverse polynomial probability the random partition $ X_1\cup X_2$ satisfies $\abs{(S_1\cup S_2)\cap X_1} = \abs{(S_1\cup S_2)\cap X_2} = \ell n/2$. If this is the case, the algorithm will succeed.

  The sets $V_1, V_2$ have cardinality $\abs{V_1} = \abs{V_2} = \binom{n/2}{\ell n/2}2^{\ell n/2}$.  Computing and sorting $V_1, V_2$ thus takes time $\wbo{\binom{n/2}{\ell n/2}2^{\ell n/2}}$.  For each element $v_2\in V_2$, binary search over $V_1$ takes times logarithmic in $\abs{V_1}$.  The result follows from Fact~\ref{fact:entropy}.
\end{proof}

\begin{figure}[htbp]
  \begin{algorithm}[H]
    \caption{Classical representation technique for \SSub}
    \label{Algo:classical-Rep-Dsum}
    \KwIn{Instance $(a,s)$ of \SSub\ with $\sum_{i=1}^n a_i < 2^{4n}$ and maximum solution ratio $\ell$.}
    \KwOut{Two subsets $S_1,S_2 \subseteq [n]$.} \vspace{0.3\baselineskip}
    Set $b = 1 -\ell$ if $\ell > 1/2$ and $b=1/2$ otherwise. \\
    \Indp Choose a random prime $p\in [2^{bn} \tdots 2^{bn + 1}]$ and a random integer $k \in [0 \tdots p-1]$. \\
    Construct the table $t_p[i,j]$ for $i=0,\ldots, n$ and $j=0,\ldots, p-1$ (see Section~\ref{Sec:DynProg}). \\
    Enumerate $T_{p,k}$ and $T_{p,(k-s)\bmod p}$, and sort $T_{p,(k-s)\bmod p}$.
    For each $S_1 \in T_{p,k}$, binary search for $S_2 \in T_{p,(k-s) \bmod p}$ such that $\Sigma(S_1) = \Sigma(S_2)+s$ and $S_1 \neq S_2$. If found, output the pair $(S_1, S_2)$.
  \end{algorithm}
\end{figure}

\begin{theorem}[\SSub, classical representation]
  \label{Thm:classical-Rep-Dsum}
  Given an instance of \SSub\ with $\sum_{i=1}^n a_i < 2^{4n}$  and maximum solution ratio $\ell \in (0,1)$, \emph{Algorithm~\ref{Algo:classical-Rep-Dsum}} finds a solution with inverse polynomial probability in time $\wbo{2^{bn} + 2^{(1-b) n}}$, where $b= 1-\ell$ if $\ell>1/2$ and $b=1/2$ otherwise.
\end{theorem}

\begin{proof}
  The choice of $b$ satisfies $b\le 1-\ell$. By Lemmas~\ref{lem:pz} and~\ref{lem:maxRatioEsum}, with probability $\om{1/n}$ there is at least one solution pair contained in~$T_{p,k} \times T_{p,(k-s) \bmod p}$.  By Lemma~\ref{lem:size} and Markov's inequality, the sizes of $T_{p,k}$ and $T_{p,(k-s) \bmod p}$ are at most $t_{p,k}, t_{p,(k-s) \bmod p} \leq n^2 2^{(1-b) n}$ with probability at least $1-1/n^2$. Thus, with probability $\om{1/n}$ we can assume that both of these events occur. If this is the case, then enumeration and sorting of $T_{p,k}$, $T_{p,(k-s)\bmod p}$ can be completed in time $\wbo{2^{(1-b)n}}$ (Theorem~\ref{Thm:effEnum}) after constructing the table $t_p$ in time $\wbo{2^{bn}}$ (Lemma~\ref{lem:dp-table-time}).
\end{proof}

For each value of $\ell$, choosing the better of Algorithms~\ref{Algo:classical-Mim-Dsum} and~\ref{Algo:classical-Rep-Dsum} gives the following result.

\begin{theorem}[\SSub, classical]
  \label{Thm:classical-Dsum}
  There is a classical algorithm that, given an instance of \SSub\ with maximum solution ratio $\ell \in (0,1)$, outputs a solution with at least inverse polynomial probability in time $\wbo{2^{\gamma(\ell) n}}$ where
    \[\gamma(\ell) =
    \begin{cases}
      1/2   & \text{if $\ell_1 \leq \ell < 1/2$,}  \\[1mm]
      \ell  & \text{if $1/2 \leq \ell < \ell_2$,}   \\[1mm]
      (h(\ell) + \ell)/2 & \text{otherwise}
    \end{cases}\]
  and $\ell_1 \approx 0.227$ and $\ell_2\approx 0.773$ are solutions to the equations $(h(\ell) + \ell)/2 =1/2$ and $(h(\ell) + \ell)/2 = \ell$ respectively. In particular, the worst-case complexity is~$\bo{2^{0.773 n}}$.
\end{theorem}

In comparison with the above result, the algorithm of~\cite{MNPW19c} for \Equ\ has running time $\wbo{2^{\gamma(\lmin) n}}$ where $\lmin$ is the minimum solution ratio (rather than maximum), and $\gamma(\lmin) = \lmin$ for $1/2\le \lmin < \ell_2$ and $\gamma(\lmin) = (h(\lmin) + \lmin)/2$ otherwise. We do not know if a similar algorithm exists for \SSub\ based on the minimum solution ratio.

%% file: quantumEqualSum.tex
We first recall the concept of a \emph{minimum solution}, introduced in~\cite{MNPW19c}.

\begin{definition}[Minimum solution]
  \label{def:min-sol}
  Two disjoint subsets $S_1,S_2 \subseteq \{1,\ldots,n\}$ that form a solution to an instance of \Equ\ are a {\em minimum solution} if their size $\abs{S_1} + \abs{S_2} = \lmin$ is the smallest among all such solutions. We call $\lmin \in (0,1)$ the minimum solution ratio.
\end{definition}

We prove that, for the special case of \Equ\ (Problem~\ref{Pbm:Equ}), we can reformulate the results of Section~\ref{Sec:DSum} to make use of the minimum solution ratio $\lmin$ instead of the maximum one.

\begin{theorem}[\Equ, quantum]
  \label{Thm:Esum_quantum}
  There is a quantum algorithm that, given an instance of \Equ\ with minimum solution ratio $\lmin \in (0,1)$, outputs a solution with at least inverse polynomial probability in time $\wbo{2^{\gamma'(\lmin) n}}$ where
    \[\gamma'(\lmin) =
    \begin{cases}
     \frac{1}{2} - \frac{1-\lmin}{4} h\pt[\big]{\frac{\lmin}{2(1-\lmin)}} & \text{if $\ell_1' \le \lmin < 1/2$} \\[1mm]
      (1 + \lmin)/4       & \text{if $1/2 \leq \lmin \leq 3/5$,}    \\[1mm]
      \lmin/2 + 1/10      & \text{if $3/5 < \lmin < \ell_2'$,}  \\[1mm]
      (h(\lmin) + \lmin)/3 & \text{otherwise}
    \end{cases}\]
  and $\ell_1' \approx 0.273$ and $\ell_2' \approx 0.809$ are solutions to the equations $(h(\lmin) + \lmin)/3 = 1/2 - (1-\lmin)h\pt[\big]{\frac{\lmin}{2(1-\lmin)}}/4$ and $(h(\lmin) + \lmin)/3 = \lmin/2 + 1/10$ respectively. In particular, the worst-case complexity is~$\bo{2^{0.504 n}}$.
\end{theorem}

The proof follows closely that of the quantum algorithm for \SSub\ (Theorem~\ref{Thm:Dsum}), with the main difference coming from a bound on the size of the collision values set $V = \set{v \in \N : \exists S_1 \neq S_2, v = \Sigma(S_1) = \Sigma(S_2)}$ for \Equ. For \SSub, Lemma~\ref{lem:maxRatioEsum} gives the bound $\abs{V} \ge 2^{(1-\ell)n}$ in terms of the \textit{maximum solution ratio} $\ell$. For \Equ\, we can obtain a similar statement in terms of the \textit{minimum solution ratio}.

\begin{lemma}
  \label{lem:minRatioEsum}
  If an instance of \Equ\ has minimum solution ratio $\lmin$ then the collision values set $V = \set{v \in \N : \exists  S_1 \neq S_2, v = \Sigma(S_1) = \Sigma(S_2)}$ satisfies
    \[\abs{V} \geq
        \begin{cases}
          2^{(1- \lmin)n} & \text{if }  \lmin > 1/2,\\
          2^{(1-\lmin)h\left(\frac{\lmin}{2(1-\lmin)}\right)n} & \text{otherwise.}
        \end{cases}\]
\end{lemma}

\begin{proof}
  The case $\lmin > 1/2$ is dealt with in~\cite{MNPW19c}, therefore consider $\lmin \le 1/2$. Let $S_1,S_2\subseteq \{1, \ldots, n\}$ be a minimum solution of size $\lmin n$. Then for any $S \subseteq  \overline{S_1 \cup S_2}$, with $|S| = \frac{\lmin n}{2} - 1$, the sets $S \cup S_1$ and $S \cup S_2$ form a solution, and for $S \neq S'$, the values $\Sigma(S_1 \cup S) $ and $\Sigma(S_1 \cup S')$ are distinct. Indeed, if this were not the case then $S \setminus S'$ and $S' \setminus S$ would form a disjoint solution of size less than $\lmin n$. Therefore $|V| \geq \binom{n(1-\lmin)}{\frac{\lmin n}{2 } -1}$, and the statement follows from Fact~\ref{fact:entropy}.
\end{proof}

We can now prove Theorem~\ref{Thm:Esum_quantum}.

\begin{proof}
  For each minimum solution ratio $\lmin\in(0,1)$, we use the better of Algorithms~\ref{Algo:Rep-Dsum} and~\ref{Algo:Mim-Dsum} with two modifications: (i) the minimum solution ratio $\lmin$ is used in place of the maximum solution ratio $\ell$ in the input of the algorithms, and (ii) we choose the value of $b$ in step~1 of Algorithm~\ref{Algo:Rep-Dsum} to be
    \[b(\lmin) =
      \begin{cases}
         \frac{1}{2} - \frac{1-\lmin}{4} h\pt[\big]{\frac{\lmin}{2(1-\lmin)}} & {\rm ~if~} \lmin \leq 1/2, \\
        (1 + \lmin)/4 & {\rm ~if~} 1/2 < \lmin \leq 3/5, \\
        2/5  &  {\rm ~if~} 3/5 < \lmin.
      \end{cases}\]
  The analysis of Algorithm~\ref{Algo:Mim-Dsum} is unaffected by the change to minimum solution ratio, as is the analysis of Algorithm~\ref{Algo:Rep-Dsum} for $\lmin > 1/2$. For Algorithm~\ref{Algo:Rep-Dsum} and $\lmin\le 1/2$, repeating the analysis of Lemma~\ref{lem:pz} using $\abs{V}\ge 2^{(1-\lmin)h\pt[\big]{\frac{\lmin}{2(1-\lmin)}}n}$ gives
    $\Pr_{p,k}\bc{v_{p,k} \geq 2^{ (z-b)n - 2} } = \om{1/n}$
  where $z = (1-\lmin)h\pt[\big]{\frac{\lmin}{2(1-\lmin)}}$. Recalling the proof of Theorem~\ref{Thm:Rep-Dsum}, the construction of the dynamic programming table takes time $\wbo{2^{bn}}$, and a collision can be found in time $\wbo{t_{p,k}^{2/3}/v_{p,k}^{1/3}} = \wbo{2^{(2-z-b)/3}}$. The running time for $\lmin < 1/2$ follows from balancing these two costs, i.e. by setting $b = (2-z-b)/3$.
\end{proof}